\documentclass{article}

\usepackage[colorlinks=true,linkcolor=blue,citecolor=blue]{hyperref}

\usepackage[notref,notcite]{showkeys}
\usepackage{color,array,colortbl}

\usepackage{graphicx}
\usepackage{tikz, pgfplots}
\usetikzlibrary[pgfplots.groupplots]
\usetikzlibrary{patterns}
\usepackage{filecontents}
\usepackage{color}
\usepackage{amsmath,amsfonts,amsthm,bm}
\usepackage[a4paper]{geometry}
\usepackage{algorithm,algorithmic}
\usepackage{placeins}
\usepackage{natbib}
\usepackage{multirow}
\usepackage{subfig}
\usepackage{pdflscape}

\usepackage{alltt}
\usepackage{fullpage}
\definecolor{string}{rgb}{0.7,0.0,0.0}
\definecolor{comment}{rgb}{0.13,0.54,0.13}
\definecolor{keyword}{rgb}{0.0,0.0,1.0}

\newcolumntype{A}{>{\columncolor[gray]{0.8}}c}

\newcommand{\bszero}{\boldsymbol{0}} 
\newcommand{\bsone}{\boldsymbol{1}}  
\newcommand{\bsc}{\boldsymbol{c}}    
\newcommand{\bsv}{\boldsymbol{v}}    
\newcommand{\bsw}{\boldsymbol{w}}    
\newcommand{\bsx}{\boldsymbol{x}}    
\newcommand{\bsy}{\boldsymbol{y}}    
\newcommand{\bsS}{\boldsymbol{S}}    

\newcommand{\E}{\mathrm{e}} 
\newcommand{\EE}{\mathbb{E}} 

\def\RR{\mathbb R}
\def\FF{\mathbb F}

\def\mF{\mathcal F}

\def\mL{\mathcal L}

\theoremstyle{plain}

  \newtheorem{corollary}{Corollary}
  \newtheorem{proposition}{Proposition}
  
\theoremstyle{definition}

\theoremstyle{remark}

\newcommand{\RefEq}[1]{~\textup{(\ref{#1})}}

\newcommand{\RefSec}[1]{Section~\textup{\ref{#1}}}

\newcommand{\RefProp}[1]{Proposition~\textup{\ref{#1}}}

\newcommand{\RefCol}[1]{Corollary~\textup{\ref{#1}}}
\newcommand{\RefAlg}[1]{Algorithm~\textup{\ref{#1}}}
\newcommand{\RefFig}[1]{Figure~\textup{\ref{#1}}}

\newcommand{\RefTab}[1]{Table~\textup{\ref{#1}}}

\title{A Monte Carlo method for optimal portfolio executions}
\author{Nico Achtsis and Dirk Nuyens}

\providecommand{\keywords}[1]{\textbf{\textit{Key words.}} #1}
\providecommand{\AMS}[1]{\textbf{\textit{AMS subject classifications.}} #1}

\begin{document}

\maketitle

\begin{abstract}
We treat the problem of mean-variance optimal execution in markets with limited liquidity and varying volatility. 
When the market parameters are assumed constant, an analytical solution exists for the optimal trading rate.
In general however, this problem leads to a non-linear Hamilton--Jacobi--Bellman PDE, which has to be solved numerically.
Since solving such a PDE is a complex procedure, \citet{almgren2012optimal} mentions a sub-optimal control that can be used as an approximation.
This strategy assumes the market parameters are constant, and hence takes the analytical solution from the stationary problem, but updates the strategy each time the market parameters change.
It is called the rolling horizon strategy (RHS), because it is essentially a continuously updated static control with contracting horizon.
It is easy to extend to the multi-asset case as we will show.
In this paper, we propose a rolling horizon Monte Carlo algorithm (RHMC).
Our method chooses a trading rate based on simulations using a sub-optimal control.
The potential upside of this method is that our proposed RHMC method not only uses current market information, such as the RHS, but also uses simulations to infer future market behaviour as well.
Our new method is naturally formulated for the multi-asset case and allows the freedom to choose the structure of the stochastic driver processes.
The results indicate that our method can significantly outperform the RHS.
We also provide some insights into the RHS, showing that it converges to the optimal solution for strong risk-averse traders, at least in the setting of \citet{almgren2012optimal}.
\end{abstract}

\keywords{Quasi-Monte Carlo (QMC), Optimal asset execution, Optimal control problems.}

\AMS{91G60, 91G80}

\section{Related literature and overview}
This paper revolves around a fundamental part of algorithmic trading, namely, trade scheduling. 
When facing the execution of a large block of assets over a fixed time interval, it is usually beneficial to split up the order in several smaller blocks over the time interval to reduce market impact.
Finding the optimum schedule requires balancing between market risk and liquidity risk.
The former entails the risk of adverse price moves of the assets that are traded, for this reason slower trading will lead to higher market risk.
Liquidity risk on the other hand corresponds to the difference in the pretrade price before the order is executed and the actual execution price.
This difference is called slippage, and the accumulation of these gaps in prices will be larger when trading is faster.
Consequently trade scheduling poses a dilemma between trading fast to eliminate market risk against trading slow to minimize slippage.
We formulate this so-called optimal execution problem in terms of a standard mean-variance optimization.

The first market impact models were based on the discrete-time models constructed by \citet{bertsimas1998optimal} and \citet{almgren2001optimal} and their continuous-time variants proposed by
\citet{almgren2003optimal}. 
These models all separated the impact into two components: an instantaneous one affecting only the individual trades that also triggered it, and a permanent effect that has an impact on all future trades.
Research in market microstructure however suggests that market impact decays over time, as one can see in the overview of, e.g., \citet{eisler2012price}. 
The first models to pick up on this fact are those of \citet{obizhaeva2005optimal} and \citet{potters2003more}. 
The latter is an example of a limit order book model, meaning the authors model the dynamics of supply and demand in the order book to find the optimal execution algorithm. 
This model was further developed by \citet{alfonsi2010optimal,alfonsi2009order}, in particular to include nonlinear price impacts. 

An important aspect of modelling the market impact is consistency. \citet{gatheral2010no} provides a good overview of what properties are desirable for such a model; e.g., one should not allow an algorithm to manipulate the market through trading in order to profit. 
\citet{huberman2004price} were among the first to point out that it is not sufficient to require the absence of arbitrage strategies in the usual sense. 
They illustrated that the feedback of trading strategies can lead to so-called price manipulation strategies that, when suitably rescaled and repeated, can create a weak form of arbitrage. 
Again, \citet{gatheral2010no} explores the relation between impact functions and such weak forms of arbitrage. 
Furthermore, it was shown by \citet{alfonsi2009order} that transaction-triggered price manipulation is possible in models that do not allow for price manipulation in the sense of \citet{huberman2004price}. 
The papers by \citet{alfonsi2010optimal,alfonsi2009order} provide models without such phenomena.

The goal of this paper is to complement the work on single asset optimal execution schemes by Robert Almgren, in particular that of \citet{almgren2012optimal}. 
Therein the author assumes a trader perceives the asset price plus an instantaneous effect based on the current trading speed.
The problem of optimal execution is there solved under the assumption that volatility and liquidity vary perfectly inversely (termed coordinated variation). 
In this paper we extend the model to multiple assets and do not need the assumption of coordinated variation.
As a minor side step, in \citet[Section~1.4]{almgren2012optimal} a simplified solution called the rolling horizon strategy (RHS) is also proposed.
This strategy entails that a stationary solution is sought, which is updated during the trading period. 
In general this will not be optimal, but it is argued to be easy to implement while providing a reasonable solution. 
In this paper, we propose a rolling horizon Monte Carlo algorithm (RHMC) in which the trading rates are calculated based on simulations using a sub-optimal control. 
In this way our new RHMC method not only uses current market information, such as the RHS, but also uses simulations to infer future market behaviour.

\citet{almgren2012optimal} deals with optimal execution within a mean-variance framework, meaning that not only the expected cost of trading is considered, but also the risk profile of the trader. 
Other papers usually deal with execution assuming a risk-neutral trader (and hence only look at the expected cost incurred from trading).
There are some papers which have already made multi-asset extensions to the Almgren framework, see for instance \citet{konishiy2001optimal} and \citet{schoneborn2011adaptive}.
These papers however only deal with constant liquidity and volatility, and assume that there is no cross-liquidity effect between assets.

\bigskip

In \RefSec{sec:EP} we give an overview of the optimal execution problem.
Under the assumption that market parameters are constant, we derive the optimal strategy in \RefSec{sec:CC}.
In \RefSec{sec:DynProb} we look at strategies which allow dynamic market parameters.
\RefSec{sec:RHA} contains two results on the RHS, showing when this strategy becomes optimal as well as its behaviour for risk-neutral traders.
The main contribution is in \RefSec{sec:MC}, where we explain our rolling horizon Monte Carlo method (RHMC).
This method is then tested numerically in \RefSec{sec:NumResults}.
Finally we conclude in \RefSec{sec:end}.

\section{Execution problem}\label{sec:EP}
The execution problem consists in liquidating or acquiring $n$ asset positions over a finite trading period $[0,T]$, $T<\infty$. 
We use the following notation: let $x_i(t)$ denote the number of $i$th shares that still need to be bought or sold at time $t$.
Then $x_i(0)$ equals the number of shares that need to be traded at the inception of the program, where $x_i(0)<0$ corresponds to a buying program, and $x_i(0)>0$ to a selling program.
In any case the program terminates at time $T$ with $x_i(T)=0$.
We will denote the initial position by the vector $\bsx_0$.
The trading speed, or first derivative of $x_i$, will be denoted by $v_i$.
We can choose either $x_i$ or $v_i$ as the control variable, since both determine the other.
However, for numerical stability we will formulate our algorithms in terms of $\bsx$.
We use the notation $x_i''$ to denote the second derivative of $x_i$ with respect to time $t$.

\subsection{Model}\label{sec:Model}
We consider a probability space $(\Omega,\mF,P)$ endowed with a filtration $\FF=(\mF(t))_{t\in[0,T]}$ which represents the information structure available to the agent. 
We assume that $\mF(0)$ is trivial and that $\mF(T)=\mF$.  
We also suppose that $\FF$ satisfies the usual conditions of right-continuity and completeness (see, e.g., \citet{karatzas1991brownian}). 
All components of the model will be defined on the filtered probability space $(\Omega,\mF,\FF,P)$.

\subsubsection*{Risky asset}
The price of the risky assets follow an arithmetic Brownian motion
\begin{align}
 S_k(t)
 &=
 S_k(0)+\int_0^t\sigma_k(s)dW_k(s)
 \label{eq:S_dynamics}
\end{align}
with $\sigma_k(t)$ a (stochastic) function of time and $dW_k(t)dW_{\ell}(t)=\rho_{k\ell} dt$. 
The process $W_k(t)$ is a standard $\mF(t)$-adapted Brownian motion. This model is known as Bachelier's model and is widely used in the optimal execution literature. 
See, among others, \citet{alfonsi2009order,almgren2012optimal,bertsimas1998optimal}. We prefer this model because of its simplicity, and standardised nature.
We will write $\bsS(t)$ for the vector containing $S_1(t)$ to $S_n(t)$.
Under the Bachelier model the dynamics for $\bsS(t)$ imply
\begin{align*}
 \bsS(T)
 &\sim
 N\left(\bsS(0),\int_0^T\Sigma(s)ds\right),
\end{align*}
where the matrix $\Sigma(s)$ consists of the elements $\sigma_k(s)\sigma_{\ell}(s)\rho_{k\ell}$.
The natural assumption is made that the matrix $\Sigma(s)$ is positive definite.
The Bachelier model could lead to negative asset values; however, since $T$ is typically small, the probability of negative asset values is negligible.

\subsubsection*{Price impact function}
The price impact function gives the price change relative to the risky asset prices $\bsS$, which from now on we refer to as the unaffected asset prices, depending on order size and market conditions. 
These impact functions have been well-studied in the empirical literature, see, e.g.,  \citet{easley1987price},  \citet{kyle1985continuous} and the theoretical literature, see, e.g., \citet{bertsimas1998optimal}. 
We will assume the price impact function from \citet{almgren2012optimal}, extended to multiple assets.
In this framework the perceived asset prices, denoted by $\tilde{\bsS}$, are
\begin{align}
 \tilde{\bsS}(t)
 &=
 \bsS(t) + \Xi(t) \bsv(t),
 \label{eq:perc_asset}
\end{align}
where 
\begin{align}
 \Xi(t)
 &=
 \begin{pmatrix}
  \eta_{11}(t) & \eta_{21}(t) & \cdots & \eta_{n1}(t) \\
  \eta_{21}(t) & \eta_{22}(t) & \cdots & \eta_{n2}(t) \\
  \vdots & \vdots & \ddots & \vdots \\
  \eta_{n1}(t) & \eta_{n2}(t) & \cdots & \eta_{nn}(t)
 \end{pmatrix}
\end{align}
is the matrix containing the (stochastic) instantaneous market impact coefficients $\eta_{ij}(t)\geq0$.
We will assume that the diagonal is strictly positive, $\eta_{ii}(t)>0$.
In our framework we assume that trading has no impact on the market impact coefficients, i.e., $\eta_{ij}(t)$ moves independently of $\bsv$.
We want $\Xi(t)$ to be positive definite for reasons explained below, which means it is a symmetric matrix.
Under this model a higher impact for a given trading rate corresponds to a less liquid asset, as the trade will eat up more of the order book (which is then instantaneously replenished).
On the other hand, a lower impact corresponds to a more liquid asset.
Therefore, we will call the $\eta$ liquidity parameters, which is more concise.

In this paper we will assume that both the $n$ volatility and $n(n+1)/2$ liquidity processes are driven by $n(n+3)/2$ correlated Ornstein--Uhlenbeck processes, which are of the form
\begin{align*}
 d\xi^k(t)
 &=
 -\frac{\xi^k(t)}{\delta_k}dt+\frac{\beta_k}{\sqrt{\delta_k}}dB^k(t),
 \qquad\qquad j=1,\ldots,\frac{n(n+3)}{2}.
\end{align*}
Here $\delta_k$ is the market relaxation time and $\beta_k$ describes the dispersion of volatility ($k=1,\ldots,n$) and liquidity ($k=n+1,\ldots,n(n+3)/2$) around their average levels.
The Brownian motions are correlated as $dB^k(t)dB^m(t)=\varrho_{km} dt$, and there is no correlation between the processes $B$ and $W$.
The volatilities depend on these processes as 
\begin{align*}
 \sigma_k(t)
 &=
 \bar{\sigma}_ke^{\xi^k(t)}, 
 \qquad\qquad
  k=1,\ldots,n,
\end{align*}
where $\bar{\sigma}_k$ is the average level of the $k$th volatility.
The liquidity parameters depend similarly on the processes as
\begin{align*}
 \eta_{k\ell}(t)
 &=
 \bar{\eta}_{k\ell}e^{\xi^{n+k+(\ell-1)n}(t)}, 
 \qquad\qquad
  k=1,\ldots,n,\,\,\,\ell=1,\ldots,k,
\end{align*}
where $\bar{\eta}_{k\ell}$ is the average level of the liquidity process $\eta_{k\ell}$.
The choice of the driving stochastic processes is motivated by  \citet{almgren2012optimal}, however, all our results are still valid for a different choice of driving processes.

\subsubsection*{Coordinated variation}
In order to reduce the dimensionality, \citet{almgren2012optimal} (which only considers $n=1$) assumes that $\eta(t)$ and $\sigma^2(t)$ vary perfectly inversely, i.e.,
\begin{align*}
 \sigma^2(t)\eta(t)
 &= 
 \bar{\sigma}^2\bar{\eta},
\end{align*}
or equivalently,
\begin{align*}
 \xi^1(t)
 &=
 -\frac{1}{2}\xi^2(t).
\end{align*}
In terms of the dynamics of the processes, coordinated variation corresponds to setting $\delta_1=\delta_2$, $\beta_2-2\beta_1=0$, $\rho=-1$ and $\xi^1(0)=\xi^2(0)$.
This relationship is argued to be a natural consequence of a trading time model in which the single source of uncertainty is the arrival rate of trade events. 
In such a model each trade event brings a fixed amount of price variance and the opportunity to trade a fixed number of shares for a particular cost simultaneously.
It is argued in the same paper that this assumption can be seriously violated during events when volatility sharply increases while liquidity is withdrawn simultaneously. 
We will therefore not make this assumption.

\subsection{Cost of trading}\label{sec:CoT}
The cost of trading given a control $\bsv$, denoted by $C$, is the difference between the amount paid to trade the assets and its initial market value $\bsx_0^T\bsS(0)$.
By using partial integration for c\`adl\`ag processes we find
\begin{align*}
 C
 &=
 \int_0^T \bsv^T(t)\tilde{\bsS}(t)dt  -\bsx_0^T\bsS(0) \\
 &=
 \sum_{k=1}^n\int_0^T \sigma_k(t)x_k(t)dW_k(t) + \int_0^T \bsv^T(t)\Xi(t)\bsv(t)dt.
\end{align*}
We will determine the optimal control $\bsv$ by the mean-variance criterion
\begin{align}
	&
	\min_{\bsv}\left[\mathbb{E}(C) + \lambda \mathrm{Var}(C)\right]\ ,	
	\label{MVCriterion}
\end{align}
where $\lambda\geq 0$ is a risk-aversion coefficient.
Note that $\lambda=0$ corresponds to a risk-neutral trader.
We need to calculate the variance term in (\ref{MVCriterion}).
Strictly speaking, it involves contributions from $W_k(t)$, the uncertainty in the asset price, as well as from $\Xi(t)$ and $\Sigma(t)$, the uncertainties in the market condition.
We can circumvent the need to approximate the contributions of $\Xi(t)$ and $\Sigma(t)$ by making the so-called small-impact approximation (see \citet{almgren2012optimal}): the variance comes primarily from the price volatility represented by $\Sigma$, with lesser contributions from the uncertainty in $\Xi(t)$ and $\bsv(t)$,
\begin{align*}
	 \mathrm{Var}(C)
	 &\approx
	 \mathbb{E}\int_0^T \bsx^T(t)\Sigma(t)\bsx(t)dt
	 .
\end{align*}
This is true if the portfolio is small enough such that price changes due to impact of trading are small compared to volatility.
Under this condition the mean-variance cost function for a control $\bsv(t)$ is
\begin{align}\label{eq:CCcost}
 \mathbb{E}(C) + \lambda \mathrm{Var}(C)
 &\approx
 \mathbb{E}\int_0^T \left[ \bsv^T(t)\Xi(t)\bsv(t) + \lambda\bsx^T(t)\Sigma(t)\bsx(t) \right]dt. 
\end{align}
The object of optimal asset execution is then finding the trading rate $\bsv$ such that the above cost function is minimized.
In general, starting at some time $t\ge0$ with $x(t)$ assets left to trade, we take as the value function
\begin{align}\label{eq:OptimizProblem}
 c(t,x,\Xi,\Sigma)
 &=
 \min_{\bsv(s),\,t\leq s\leq T}\mathbb{E}\int_t^T \left[ \bsv^T(s)\Xi(s)\bsv(s) + \lambda\bsx^T(s)\Sigma(s)\bsx(s) \right]ds. 
\end{align}
Since both $\Xi(s)$ and $\Sigma(s)$ are positive definite, the cost of trading will be positive.


\section{The static problem}\label{sec:CC}

In this section, we will only consider the static problem, i.e., the case when $\Xi(t)$ and $\Sigma(t)$ are constant over time.
We will first show that this problem has a unique minimizer.
\begin{proposition}
	Assume that $\Xi$ and $\Sigma$ are two positive definite matrices.
	Then the optimization problem
	\begin{align*}
		\min_{\bsv(t),0\leq t\leq T} \int_0^T \mL(t,\bsx,\bsv) dt
	\end{align*}
	where
	\begin{align*}
 		\mL(t,\bsx,\bsv)
 		&=
 		\bsv^T(t)\Xi\bsv(t) + \lambda\bsx^T(t)\Sigma\bsx(t),
\end{align*}
	has a unique minimizer in $W^{1,2}([0,T],\RR^n)$
\end{proposition}
\begin{proof}
	It is clear that $\mL(t,\bsx,\bsv)$ is convex in $\bsv$, since $\bsv^T(t)\Xi\bsv(t)$ is a quadratic form with $\Xi$ positive definite.
	Furthermore, since $\Xi$ and $\Sigma$ are two positive definite matrices 
	\begin{align*}
		\mL(t,\bsx,\bsv)
		&\geq
		\alpha_1 \bsv^T(t)\bsv(t) + \lambda\alpha_2 \bsx^T(t)\bsx(t),
	\end{align*}
	where $\alpha_1$ and $\alpha_2$ are the smallest eigenvalues of $\Xi$ and $\Sigma$ respectively. 
\end{proof}
The Euler--Lagrange equations are
\begin{align*}
	\frac{\partial\mL}{\partial x_k}-\frac{d}{dt}\frac{\partial\mL}{\partial x'_k}
	&=
	0, 
	&
	 k=1,\ldots,n,
\end{align*}
where the total derivative equals
\begin{align*}
	\frac{d}{dt}\frac{\partial\mL}{\partial x'_k}
	&=
	\frac{\partial}{\partial t}\frac{\partial\mL}{\partial x_k'} + \sum_{\ell=1}^n x'_{\ell}\frac{\partial}{\partial x_{\ell}}\frac{\partial\mL}{\partial x_k'} + \sum_{\ell=1}^n x''_{\ell}\frac{\partial}{\partial x'_{\ell}}\frac{\partial\mL}{\partial x_k'}, \qquad k=1,\ldots,n.
\end{align*}
Straightforward calculations then show that the Euler--Lagrange equations reduce to
the system 
\begin{align}\label{eq:EL_eqns}
 \Xi \bsx''(t)
 &=
 \lambda \Sigma \, \bsx(t).
\end{align}
\begin{proposition}\label{prop:CC}
  Under the assumption that $\Xi(t) = \Xi$ and $\Sigma(t) = \Sigma$ are constant for $0 \le t \le T$, the boundary value problem
  \begin{align}\label{eq:sysODE}
    \bsx''(t)
    &=
    \lambda \Xi^{-1} \Sigma \, \bsx(t)
    ,
    \quad
    \text{for } 0 \le t \le T
    ,
    \quad
    \text{ where }
    \bsx(0) = \bsx_0
    \text{ and }
    \bsx(T) = \bszero
    ,
  \end{align}
  has the solution, for $\lambda > 0$, 
  \begin{align}\label{eq:xCC}
    \bsx(t)
    &=
    \Omega(t,T,\Xi,\Sigma) \, \bsx_0
    ,
    &\text{where }
    \Omega(t,T,\Xi,\Sigma) 
    &=
    \sinh(C (T-t)) \sinh(C T)^{-1}
    ,
  \end{align}
  and
  \begin{align}\label{eq:vCC}
    \bsv(t)
    &=
    \Omega'(t,T,\Xi,\Sigma) \, \bsx_0
    ,
    &\text{where }
    \Omega'(t,T,\Xi,\Sigma) 
    &=
    -\cosh(C(T-t)) \sinh(C T)^{-1} C
    ,
  \end{align}
  where $C$ is a matrix square root such that $C^2 = \lambda \Xi^{-1} \Sigma$ and $\sinh$ and $\cosh$ are matrix functions, i.e., $\sinh(A) = \sum_{k \ge 0} A^{2k+1} / (2k+1)!$ and $\cosh(A) = \sum_{k \ge 0} A^{2k} / (2k)!$, and $\sinh(C T)^{-1}$ is the matrix inverse of $\sinh(C T)$.
\end{proposition}
\begin{proof}
  Set $B = \lambda \Xi^{-1} \Sigma$.
  This matrix has full rank by assumption and thus we can find a square root $C = B^{1/2}$ such that $C^2 = B$.
  
  Now introduce $\bsy(t) = ( \bsx(t) , \bsx'(t) )$ then we obtain a first order system with constant coefficients
\begin{align}\label{eq:dy}
  \bsy'(t)
  &=
  \underbrace{\begin{pmatrix} 0 & I_n \\ B & 0 \end{pmatrix}}_{=: A} \bsy(t)
  =
  A \, \bsy(t)
\end{align}
for which the fundamental matrix of solutions is given by the matrix exponential
\begin{align*}
  \Psi(t)
  &=
  \E^{A t}
  =
  \sum_{k \ge 0} \frac{A^k \, t^k}{k!}
  .
\end{align*}
We have
\begin{align*}
  A^k
  &=
  \begin{cases}
    \begin{pmatrix} B^{k/2} & 0 \\ 0 & B^{k/2} \end{pmatrix}          & \text{for } k = 0, 2, 4, \ldots , \\[5mm]
    \begin{pmatrix} 0 & B^{(k-1)/2} \\ B^{(k+1)/2} & 0 \end{pmatrix}  & \text{for } k = 1, 3, 5, \ldots .
  \end{cases}
\end{align*}
Now using $\cosh(A) = \sum_{k \ge 0} A^{2k} / (2k)!$ and $\sinh(A) = \sum_{k \ge 0} A^{2k+1} / (2k+1)!$, we can write, with $C = B^{1/2}$,
\begin{align*}
  \Psi(t)
  &=
  \begin{pmatrix}
    \cosh(C t) & \sinh(C t) C^{-1} \\[1mm]
    \sinh(C t) C & \cosh(C t) 
  \end{pmatrix}
  ,
\end{align*}
and the general solution to\RefEq{eq:dy} is thus given by
\begin{align*}
  \bsy(t)
  &=
  \Psi(t) \, \bsc
  .
\end{align*}
Write $\bsc = (\bsc_1, \bsc_2)$ then the boundary conditions of the original problem\RefEq{eq:sysODE} give
\begin{align*}
  \begin{cases}
    \bsx(0) = \cosh(C 0) \, \bsc_1 + \sinh(C 0) C^{-1} \, \bsc_2 &= \bsx_0 \\
    \bsx(T) = \cosh(C T) \, \bsc_1 + \sinh(C T) C^{-1} \, \bsc_2 &= \bszero \\
  \end{cases}
\end{align*}
from which it immediately follows that $\bsc_1 = \bsx_0$ and $\bsc_2 = - C \sinh(CT)^{-1} \cosh(C T) \, \bsx_0$.
Here $\sinh(CT)^{-1}$ means the matrix inverse of $\sinh(CT)$. 
The solution to\RefEq{eq:sysODE} is thus given by
\begin{align*}
  \bsx(t)
  &=
  \left( \cosh(C t) - \sinh(C t) (\sinh(CT))^{-1} \cosh(C T) \right) \bsx_0
  \\
  &=
  \sinh(C (T-t)) \sinh(C T)^{-1} \bsx_0
  ,
\end{align*}
where we used the fact that $\sinh(CT)^{-1}$ commutes with $\cosh(CT)$ assuming $C$ has full rank, which follows from $B$ having full rank.
That the two matrices commute follows easily by using the Taylor series for $\sinh(A)$ and $\cosh(A)$ and the eigenvalue decomposition of $A=CT=V \Lambda V^{-1}$, i.e., $\cosh(A) \sinh(A)^{-1} = V \cosh(\Lambda) V^{-1} V \sinh(\Lambda)^{-1}V^{-1}$ and the diagonal matrices $\sinh(\Lambda)^{-1}$ and $\cosh(\Lambda)$ commute.
\end{proof}

The above result was formulated for the interval $0 \le t \le T$.
The obvious change for $t \le s \le T$ gives, with $\bsx(t) = \bsx_{t}$ given,
\begin{align}\label{eq:CC}
    \bsx(s)
    &=
    \Omega(s-t,T-t,\Xi,\Sigma) \, \bsx_t
    ,
    &\text{and}&&
    \bsv(s)
    &=
    \Omega'(s-t,T-t,\Xi,\Sigma) \, \bsx_t
    ,
\end{align}
where as $\Xi(s)$ and $\Omega(s)$ are constant we might think of $\Xi = \Xi(t)$ and $\Sigma = \Sigma(t)$.
This will become of use in the following.

These strategies do not adapt to fluctuations in market parameters, and we will denote them by CC, short for {\it constant coefficients}.
For the cases $n=1$ and $n=2$, we have the following two corollaries.

\begin{corollary}\label{cor:1A}
 Assuming $n=1$, i.e., there is only trading in one asset, and $\eta(t) = \eta$ and $\sigma(t) = \sigma$ constant, the optimal trading trajectory for $0 \le t \le T$ is given by
\begin{align*}
 x(t)
 &=
 \frac{\sinh(\mu(T-t))}{\sinh(\mu T)} x_0,
 &\text{and}&&
 v(t)
 &=
 \frac{-\cosh(\mu(T-t))}{\sinh(\mu T)} \mu \, x_0
 ,
\end{align*}
with $\mu^2=\lambda\sigma^2/\eta$.
\end{corollary}

\begin{corollary}\label{cor:2A}
 When trading in two assets, i.e., $n=2$, and $\Xi(t) = \Xi$ and $\Sigma(t) = \Sigma$ constant, the optimal trading trajectories for $0 \le t \le T$ are given by
\begin{align}\label{eq:optvCC2A}
 \bsx(t)
 &=
 \frac1{\theta_1 - \theta_2}
  \begin{pmatrix}
    \theta_1 s_1(t)-\theta_2 s_2(t)  & \theta_1 \theta_2 (s_2(t)-s_1(t)) \\[2mm]
    s_1(t)-s_2(t) & \theta_1 s_2(t)-\theta_2 s_1(t)
  \end{pmatrix}
  \bsx_0
  ,
\end{align}
and
\begin{align*}
 \bsv(t)
 &=
 \frac1{\theta_1 - \theta_2}
  \begin{pmatrix}
    \theta_1 s'_1(t)-\theta_2 s'_2(t)  & \theta_1 \theta_2 (s'_2(t)-s'_1(t)) \\[2mm]
    s'_1(t)-s'_2(t) & \theta_1 s'_2(t)-\theta_2 s'_1(t)
  \end{pmatrix}
  \bsx_0
  ,
\end{align*}
where for
\begin{align*}
  2 \lambda \Xi^{-1} \Sigma
  &=
  \begin{pmatrix}
    a & b \\ c & d 
  \end{pmatrix}
  \\
  &=
  \frac{2\lambda}{\eta_{12}^2+2 \eta_{12} \eta_{21}+\eta_{21}^2-4 \eta_{11} \eta_{22}}
  \begin{pmatrix}
 \sigma_1 \sigma_2 \rho (\eta_{12}+\eta_{21}) - 2 \sigma_1^2 \eta_{22} &
 \sigma_2^2 (\eta_{12}+\eta_{21})-2 \rho \sigma_1 \sigma_2 \eta_{22} \\
 \sigma_1^2 (\eta_{12}+\eta_{21})-2 \rho \sigma_1 \sigma_2 \eta_{11} & 
 \sigma_1 \sigma_2 \rho (\eta_{12}+\eta_{21})-2 \sigma_2^2 \eta_{11}
  \end{pmatrix}
\end{align*}
we set $D = a^2+4 b c-2 a d+d^2$ and
\begin{align*}
  \alpha &= a - d, 
  &
  \beta &= a + d,
  \\
  \mu_1^2 &= \frac{\beta - \sqrt D}{2},
  &
  \mu_2^2 &= \frac{\beta + \sqrt D}{2},
  \\
  \theta_1 &= \frac{\alpha - \sqrt D}{2c},
  &
  \theta_2 &= \frac{\alpha + \sqrt D}{2c},
  \\
  s_1(t) &= \frac{\sinh(\mu_1(T-t))}{\sinh(\mu_1 T)}, 
  &
  s_2(t) &= \frac{\sinh(\mu_2(T-t))}{\sinh(\mu_2 T)},
  \\
  s'_1(t) &= -\mu_1 \frac{\cosh(\mu_1(T-t))}{\sinh(\mu_1 T)}, 
  &
  s'_2(t) &= -\mu_2 \frac{\cosh(\mu_2(T-t))}{\sinh(\mu_2 T)}.
\end{align*}
\end{corollary}
\begin{proof}
  The proof follows by making an eigenvalue decomposition of $B = \lambda \Xi^{-1} \Sigma = V \Lambda V^{-1}$, then $C = V \Lambda^{1/2} V^{-1}$ and $\sinh(C(T-t)) = V \sinh(\Lambda^{1/2}(T-t)) V^{-1}$ and $\sinh(CT)^{-1} = V \sinh(\Lambda^{1/2}T)^{-1} V^{-1}$.
  For $2\times 2$ matrices this can all be done analytically and one obtains the given result.
\end{proof}

Finally, we have the following corollary.
\begin{corollary}\label{cor:CC}
 Assume that $\Xi(s) = \Xi$ and $\Sigma(s) = \Sigma$ are constant over $t \le s \le T$.
 Denote with $\bsx(s)$ the optimal trading strategy for $t \le s \le T$ found using\RefEq{eq:EL_eqns}, assuming the asset level at time $s=t$ to be $\bsx_t$.
 As a consequence of \RefProp{prop:CC}, the cost function\RefEq{eq:CCcost} over $[t,T]$ is a quadratic polynomial in the elements of $\bsx_t$.
\end{corollary}
\begin{proof}
Due to \RefProp{prop:CC}, the cost function for this strategy, with $\Omega(s) = \Omega(s-t,T-t,\Xi,\Sigma)$ and $\Omega'(s) = \Omega'(s-t,T-t,\Xi,\Sigma)$, is given by
\begin{align*}
 \int_t^T \left[ \bsv^T(s)\Xi\bsv(s) + \lambda \, \bsx^T(s)\Sigma\bsx(s) \right] ds
 &=
 \bsx_t^T \left( \int_t^T \left[ \Omega'(s)^T \Xi \Omega'(s) + \lambda \, \Omega(s)^T \Sigma \Omega(s) \right] ds \right) \bsx_t
 =
 \bsx_t^T Q \, \bsx_t
\end{align*}
and is a sum of quadratic expressions in $\bsx_t$ with coefficients that are integrals independent of $\bsx_t$.
Therefore the cost function when assuming constant $\Xi$ and $\Sigma$ is a quadratic polynomial in the initial asset position $\bsx(t)$.
\end{proof}


\section{The dynamic problem}\label{sec:DynProb}

This section deals with the more general situation where both $\Xi(t)$ and $\Sigma(t)$ move stochastically over time.
In this case, the Bellman principle can be used on the value function\RefEq{eq:OptimizProblem}, 
\begin{align}
 c(t,x,\Xi,\Sigma)
 &=
 \min_{\bsv(t)}\left[\bsv^T(t)\Xi(t)\bsv(t)dt + \lambda\bsx^T(t)\Sigma(t)\bsx(t)dt + \mathbb{E}c(t+dt,x+dx,\Xi+d\Xi,\Sigma+d\Sigma)\right].
 \label{eq:optprob}
\end{align}
Unfortunately finding an analytic solution to the problem is impossible in a general setting, and numerical techniques have to be used.
The case of one asset under coordinated variation, studied in \citet{almgren2012optimal}, reduces to a PDE with one spatial dimension.
Solving this problem can be done efficiently using standard techniques, and was done using a finite difference scheme in the aforementioned paper.
Increasing the number of assets and relinquishing the coordinated variation condition quickly results in a multidimensional problem that is impractical for finite difference techniques (see also \citet{longstaff2001valuing}).

\subsection{Rolling horizon strategy}\label{sec:RHA}

The so-called rolling horizon strategy (RHS) proposed in \citet{almgren2012optimal} offers a dynamic, but suboptimal trading strategy that does not need any numerical algorithm to compute.
The idea is to plug in the instantaneous values of of $\Xi(t)$ and $\Sigma(t)$ in the static solution, i.e., the solution to\RefEq{eq:EL_eqns}.
From\RefEq{eq:CC} we can write the instantaneous trading rate under RHS as $\bsv(t)=\Omega'(0,T-t,\Xi(t),\Sigma(t))\,\,\bsx(t)$.
The algorithm therefore assumes that the current market parameters will remain constant over the remainder of the program. 
When they change, the trading speed is altered using the new values.
The author argues that this strategy is strictly optimal only in the infinite-horizon case, and only when the market parameters covary in the appropriate way (but without specifying how).
It is furthermore claimed to provide a reasonable approximation, that is easy to implement.

We start by proving the following proposition which states under what condition the RHS reduces to the static solution.
\begin{proposition}
 Consider the execution problem for $n$ assets with initial position $\bsx_0$.
 If $\lambda\Xi^{-1}\Sigma\rightarrow 0_n$, with $0_n$ the $n\times n$ zero matrix, then the RHS converges to the CC solution.
 \label{prop:RHS2}
\end{proposition}
\begin{proof}
 When $\lambda\Xi^{-1}\Sigma\rightarrow 0_n$, the system of ODEs\RefEq{eq:sysODE} reduces to
 \begin{align*}
  \bsx''(t)
  &=
  0,
 \end{align*}
 which leads to the solution
 \begin{align*}
  \bsx(t)
  &=
  \bsx_0\left(1-\frac{t}{T}\right),
 \end{align*}
 i.e., the trading happens linearly independent of even the current market parameters.
 Consequently, the RHS coincides with the static solution.
\end{proof}

Conversely, the following proposition (which to our knowledge has not been studied in the literature) shows when the RHS becomes optimal, at least for the case $n=1$.
\begin{proposition}
 Assume $n=1$. Then for $\lambda\bar{\sigma}^2/\bar{\eta}\rightarrow\infty$, the rolling horizon approximation converges to the optimal solution if and only if
 \begin{align}
  \frac{1}{2}\varrho\beta_1\beta_2\sqrt{\frac{\delta_2}{\delta_1}}\delta_2-\frac{1}{2}\xi^1+\frac{1}{2}\frac{\delta_1^3}{\delta_2^2}\beta_1^2-\frac{\delta_2^2}{\delta_1}\xi^2+\frac{1}{8}\beta_2^2\delta_2
  &=
  0.
  \label{eq:RHAcond}
 \end{align}
 Note that this equation is satisfied under coordinated variation.
\label{prop:RHS}
\end{proposition}
\begin{proof}
The optimization problem\RefEq{eq:OptimizProblem} for the case $n=1$ can be written as
\begin{align*}
 c(t,x,\xi^1,\xi^2)
 &=
 \min_{\bsv(s),\,t\leq s\leq T}\mathbb{E}\int_t^T \left[ \eta(s)v^2(s) + \lambda\sigma^2(s)x^2(s) \right] ds. 
\end{align*}
The corresponding HJB equation is
\begin{align*}
 c_t-\frac{\xi^1}{\delta_1}c_{\xi^1}+\frac{\beta_1^2}{2\delta_1}c_{\xi^1\xi^1}-\frac{\xi^2}{\delta_2}c_{\xi^2}+\frac{\beta_2^2}{2\delta_2}c_{\xi^2\xi^2}+\varrho\frac{\beta_1\beta_2}{\sqrt{\delta_1\delta_2}}c_{\xi^1\xi^2}+\lambda\bar{\sigma}^2e^{\xi^1}x^2+\min_v\left[vc_x+\bar{\eta}e^{\xi^2}v^2\right] 
 &=
 0.
\end{align*}
The minimum is 
\begin{align*}
 v
 &=
 -\frac{c_xe^{-\xi^2}}{2\bar{\eta}},
\end{align*}
which means the PDE for $c$ is
\begin{align*}
 c_t-\frac{\xi^1}{\delta_1}c_{\xi^1}+\frac{\beta_1^2}{2\delta_1}c_{\xi^1\xi^1}-\frac{\xi^2}{\delta_2}c_{\xi^2}+\frac{\beta_2^2}{2\delta_2}c_{\xi^2\xi^2}+\varrho\frac{\beta_1\beta_2}{\sqrt{\delta_1\delta_2}}c_{\xi^1\xi^2}+\lambda\bar{\sigma}^2e^{\xi^1}x^2-\frac{c_x^2e^{-\xi^2}}{4\bar{\eta}} 
 &=
 0.
\end{align*}
The value function $c$ is strictly proportional to $x^2$, which means we can nondimensionalize using $\delta_1$ as the time scale and $\tau=(T-t)/\delta_1$, 
\begin{align*}
 c(t,x,\xi^1,\xi^2)
 &=
 \frac{\bar{\eta}x^2}{\delta_1}u(\tau,\xi^1,\xi^2),
\end{align*}
where $u$ is a nondimensional function of nondimensional variables.
The PDE becomes
\begin{align*}
 u_{\tau}+\xi^1u_{\xi^1}+\frac{\delta_1}{\delta_2}\xi^2u_{\xi^2}-\bar{\mu}^2\delta_1^2+e^{-\xi^1}u^2-\frac{1}{2}\beta_1^2u_{\xi^1\xi^1}-\varrho\sqrt{\frac{\delta_1}{\delta_2}}\beta_1\beta_2u_{\xi^1\xi^2}-\frac{1}{2}\beta_2^2u_{\xi^2\xi^2}
 &=
 0,
\end{align*}
where $\bar{\mu}^2=\lambda\bar{\sigma}^2/\bar{\eta}$.
We assume that this PDE has a unique solution.
The trade velocity in terms of the transformed value function is
\begin{align*}
 v
 &=
 \frac{x}{\delta_1}e^{-\xi^2}u(\tau,\xi^1,\xi^2).
\end{align*}
Using \RefCol{cor:1A} the continuous time rolling horizon strategy is given as
 \begin{align*}
  &-x\bar{\mu} e^{\xi^1-\frac{1}{2}\xi^2}\coth\left(\bar{\mu} e^{\xi^1-\frac{1}{2}\xi^2}(T-t)\right),
 \end{align*}
with $\bar{\mu}^2=\lambda\bar{\sigma}^2/\bar{\eta}$.
In terms of $u$ this gives
\begin{align*}
 u
 &=
 -\delta_1\bar{\mu} e^{2\xi^1-\frac{1}{2}\xi^2}\coth\left(\bar{\mu} e^{\xi^1-\frac{1}{2}\xi^2}(T-t)\right).
\end{align*}
Filling this in in the PDE for $u$ gives the equation
\begin{footnotesize}
\begin{align*}
 0&=
 \left(\frac{1}{\delta_1^2}\delta_2^5\beta_1^2-\frac{1}{4}\beta_2^2\delta_2^3+\varrho\beta_1\beta_2\sqrt{\delta_2}{\delta_1}\delta_2^3\right)\tau^2e^{3\xi^1-\frac{1}{2}\xi^2}\coth\left(\bar{\mu}e^{\xi^1-\frac{1}{2}\xi^2}\delta_2\tau\right)\left(1-\coth^2\left(\bar{\mu}e^{\xi^1-\frac{1}{2}\xi^2}\delta_2\tau\right)\right)\bar{\mu}^3 \\
 &+\left(-\frac{\delta_2^3}{\delta_1^2}\xi^1\tau-\delta_2^2+\frac{1}{2}\xi^2\delta_2^2\tau-\frac{1}{8}\beta_2^2\delta_2^2\tau-\delta_2^2+\frac{3\delta_2^4}{2\delta_1^2}\beta_1^2-\frac{1}{2}\varrho\beta_1\beta_2\sqrt{\frac{\delta_2}{\delta_1}}\delta_2^2\tau\right)e^{2\xi^1}\left(1-\coth^2\left(\bar{\mu}e^{\xi^1-\frac{1}{2}\xi^2}\delta_2\tau\right)\right)\bar{\mu}^2 \\ 
 &+\left(\frac{1}{2}\varrho\beta_1\beta_2\sqrt{\frac{\delta_2}{\delta_1}}\delta_2-\frac{1}{2}\xi^1+\frac{1}{2}\frac{\delta_1^3}{\delta_2^2}\beta_1^2-\frac{\delta_2^2}{\delta_1}\xi^2+\frac{1}{8}\beta_2^2\delta_2\right)e^{\xi^1+\frac{1}{2}\xi^2}\coth\left(\bar{\mu}e^{\xi^1-\frac{1}{2}\xi^2}\delta_2\tau\right)\bar{\mu}.
\end{align*}
\end{footnotesize}
It is clear that in order for this equation to hold when $\lambda\bar{\sigma}^2/\bar{\eta}\rightarrow\infty$, we need
\begin{align*}
 \frac{1}{2}\varrho\beta_1\beta_2\sqrt{\frac{\delta_2}{\delta_1}}\delta_2-\frac{1}{2}\xi^1+\frac{1}{2}\frac{\delta_1^3}{\delta_2^2}\beta_1^2-\frac{\delta_2^2}{\delta_1}\xi^2+\frac{1}{8}\beta_2^2\delta_2
 &=
 0,
\end{align*}
which is exactly\RefEq{eq:RHAcond}.
Under coordinated variation, $\delta_1=\delta_2$, $\beta_2-2\beta_1=0$, $\varrho=-1$ and $\xi^1+2\xi^2=0$, which satisfies the above equation.
\end{proof}
Numerical experiments, some of which will be shown in \RefSec{sec:2A}, implicate that $\lambda\bar{\sigma}^2/\bar{\eta}$ does not even need to be that large for the rolling horizon solution to become (nearly) optimal.
It is not clear what happens in the multi-asset case.
The results in \RefSec{sec:NumResults} seem to indicate that the RHS can still become optimal for certain parameter choices, but finding the exact condition lies outside the scope of this paper.

In the remainder of this section we will discuss the RHS in a discretized time framework, which will prove useful in the next section.
Suppose we discretize time as $(t_k)_{k=0,\ldots,M}$, using step length $\Delta t$.
From now on, the notation $\bsx_k=\bsx(t_k)$ and $\bsv_k=\bsv(t_k)$ will be used.
Assume we are at time $t_k$ with asset levels $\bsx_k$ and $\Xi(t_k)$ and $\Sigma(t_k)$ known (i.e., observed).
We need to decide on the trading speed over the interval $[t_k, t_{k+1}]$ which will be taken constant and is denoted by $\bsv_k$.
To decide on the optimal value of $\bsv_k$ we will assume that $\Xi(t_k)$ and $\Sigma(t_k)$ remain constant over the remainder of the program $[t_k,T]$ and use the CC solution. 
In fact, by using\RefEq{eq:xCC} we know the optimal value of $\bsx_{k+1}$ directly, 
\begin{align*}
  \bsx_{k+1}
  &=
  \Omega(0, T-t_k, \Sigma(t_k), \Xi(t_k)) \, \bsx_k
  =
  \Omega_k \, \bsx_k
\end{align*}
where the shorthand notation
\begin{align*}
  \Omega_k 
  &= 
  \Omega(0, T-t_k, \Sigma(t_k), \Xi(t_k))
\end{align*}
was introduced.
The trading rate can be deduced as
\begin{align*}
  \bsv_k
  &=
  \frac{\bsx_{k+1}-\bsx_k}{\Delta t}
  =
  (\Omega_k -I_n) \, \frac{\bsx_k}{\Delta t}
  .
\end{align*}
The trader will repeat this process at the next time step.

Each time step costs calculating $\Omega_k$ and a matrix-vector product.
If we assume the cost of calculating $\Omega_k$ to be $O(n^3)$, e.g., by making use of an eigen decomposition, then this dominates the cost per step.
The cost per step is thus $O(n^3)$.

If we assume we know all values of $\Xi(t_{k+\ell})$ and $\Sigma(t_{k+\ell})$ then we could propagate this as an iterative scheme and write, for $\ell > 1$,
\begin{align}\notag
  \bsx_{k+\ell}
  &=
  \Omega_{k+\ell-1}\cdots \Omega_k  \, \bsx_k
  \\\notag
  &=
  V(k,\ell,\Xi(\cdot),\Sigma(\cdot)) \, \Omega_k  \, \bsx_k
  \\\label{eq:xkl}
  &=
  V(k,\ell,\Xi(\cdot),\Sigma(\cdot)) \, \bsx_{k+1}
  ,
\end{align}
where the propagation from $\bsx_{k+1}$ to $\bsx_{k+\ell}$, $\ell \ge 1$, is given by
\begin{align*}
  V(k,\ell,\Xi(\cdot),\Sigma(\cdot))
  &=
  \prod_{m=1}^{\ell-1} \Omega_{k+\ell-m}
  .
\end{align*}
The reason to define a propagation matrix from $\bsx_{k+1}$ to $\bsx_{k+\ell}$ instead of from $\bsx_k$ to $\bsx_{k+\ell}$ has to do with the RHMC scheme which will be described next.
We note that
\begin{align}\label{eq:Vrecur}
  V(k,\ell+1,\Xi(\cdot),\Sigma(\cdot))
  &=
  \Omega_{k+\ell} \; V(k,\ell,\Xi(\cdot),\Sigma(\cdot))
  .
\end{align}
An overview of the discrete RHS method is given in \RefAlg{alg:RHS}.

The following proposition ensures the RHS scheme is numerically stable.
\begin{proposition}
	The RHS scheme\RefEq{eq:xkl} is numerically stable, i.e., all the eigenvalues of 
	\begin{align*}
		\Omega_k
		&=
		\Omega(0,T-t_k,\Sigma(t_k),\Xi(t_k))
	\end{align*}		
    are positive and bounded by one for all $t_k\in [0,T]$.
	\label{prop:stability}
\end{proposition}
\begin{proof}
	From \RefProp{prop:CC} we know that 
	\begin{align*}
		\Omega_k
		&=
		\Omega(0, T-t_k, \Sigma(t_k), \Xi(t_k)) 
		\\
		&=
		\sinh(C (T-t)) \sinh(C T)^{-1}
	\end{align*}
	where $C$ is a square root of $B=\lambda\Xi(t_k)^{-1}\Sigma(t_k)$.
	The dependence of $C$ and $B$ on $t_k$ is suppressed for ease of notation.
	 The proof follows by making an eigenvalue decomposition of $B = V \Lambda V^{-1}$, then $C = V \Lambda^{1/2} V^{-1}$ and $\sinh(C(T-t)) = V \sinh(\Lambda^{1/2}(T-t)) V^{-1}$ and $\sinh(CT)^{-1} = V \sinh(\Lambda^{1/2}T)^{-1} V^{-1}$.
	 Therefore
	 \begin{align*}
	 	\Omega_k
	 	&=
	 	V \sinh(\Lambda^{1/2}(T-t)) \sinh(\Lambda^{1/2}T)^{-1} V^{-1}
	 	,
	 \end{align*}
	 which implies that the eigenvalues of $\Omega_k$ are positive and bounded by one, since $\sinh$ is a positive and strictly increasing function on $(0,\infty)$.
	 Therefore the eigenvalues of the matrix product $V$ in\RefEq{eq:Vrecur} are also positive and bounded by one.
	 Because the time step $k$ was arbitrary, the stability holds for all time steps, leading to the conclusion that the RHS algorithm is stable over $[0,T]$.
\end{proof}

\begin{algorithm}
 \caption{The discrete RHS for optimal asset execution.}
 \label{alg:RHS}
 \begin{algorithmic}
  \STATE
  \STATE Fix a set of $M+1$ trading times, equally spaced with length $\Delta t=T/M$.
  \STATE
  \STATE At time $t_k$ observe $\Xi(t_k)$ and $\Sigma(t_k)$.
  \STATE Calculate $\bsx_{k+1}=\Omega_k \, \bsx_k$.
  \STATE Trade $\bsx_{k+1}-\bsx_k$ assets.
  \STATE
 \end{algorithmic}
\end{algorithm}

\subsection{Rolling horizon Monte Carlo method}\label{sec:MC}

We are now ready to introduce our new method: the rolling horizon Monte Carlo  (RHMC) method.
We discretize time again as $(t_k)_{k=0,\ldots,M}$, using step length $\Delta t$, 
with $\bsx_k=\bsx(t_k)$ and $\bsv_k=\bsv(t_k)$.
Using the Bellman principle the optimization problem\RefEq{eq:optprob} at time step $t_k$ becomes
\begin{align*}
 &\min_{\bsv_k} \bigg\{\bsv_k^T \Xi(t_k) \bsv_k \Delta t + \lambda\bsx_k^T\Sigma(t_k)\bsx_k\Delta t
 + \mathbb{E}\left[ c(\bsx_{k+1},\Xi(t_{k+1}),\Sigma(t_{k+1}))\left| \bsx_k,\bsv_k,\Xi(t_k),\Sigma(t_k)\right. \right]\bigg\},
\end{align*}
where $c$ is the total cost incurred from time $t_{k+1}$ onwards, assuming the control is used.
We first explain the standard technique to solve this problem in a backwards manner.
Since the expectation is conditional on the current values of $\bsx_k$, $\bsv_k$, $\Xi(t_k)$ and $\Sigma(t_k)$ one has to account for all possible values of these conditional parameters.
When using Monte Carlo, one samples $\Xi(t)$ and $\Sigma(t)$ at time steps $t_1,\ldots,t_M$.
The \emph{continuation value}, i.e., the expected value of $c$ in the optimization problem, is then approximated using a multivariate regression on powers or exponentials of the underlying variables. 
However, one has to be careful since the variable $\bsx_k$ is \emph{endogenous}: it is determined completely by the control $\bsv_k$.
As in \citet{boogert2008gas} we notice that the continuation value depends only on the asset level that is reached, not on the previous level and chosen trading rate. 
Therefore, one could discretize the endogenous variable $\bsx$ and do separate regressions for each level, depending only on the processes $\xi$. 
This scheme still needs an exponentially increasing number of regressions as the number of assets increases.
We therefore propose to use a rolling horizon Monte Carlo scheme that we will show does not require the discretization of the endogenous variable, nor the use of regressions.

While the previous methods operate \emph{backwards},
the rolling horizon Monte Carlo algorithm we propose is a \emph{forward} scheme that approximates the continuation value using a sub-optimal control.
As a major advantage we can now assume that the conditional values $\bsx_k$, $\bsv_k$, $\Xi(t_k)$ and $\Sigma(t_k)$ are known and we are only left with the evaluation of a single expectation.

\subsection*{RHMC-I: Rolling horizon Monte Carlo with RHS until the end}

Assume we are at time $t_k$ with asset levels $\bsx_k$ and $\Xi(t_k)$ and $\Sigma(t_k)$ known (i.e., observed).
We need to decide on the trading speed over the interval $[t_k, t_{k+1}]$ which will be taken constant and is denoted by $\bsv_k$.
We will try to find the optimal value of $\bsv_k$, denoted by $\bsv_k^*$, that minimizes the cost
\begin{align*}
  &c(t_k, \bsx_k, \Xi(t_k), \Sigma(t_k))
  \\
  &\qquad=
  \min_{\substack{\bsv(s) \\ t_k \le s < T}}
  \EE\left[\left.
  \int_{t_k}^T \left[ \bsv^T(s) \Xi(s) \bsv(s) + \lambda \bsx^T(s) \Sigma(s) \bsx(s) \right] ds
  \right|
  \bsx_k, \Xi(t_k), \Sigma(t_k)
  \right]
  \\
  &\qquad\approx
  \min_{\substack{\bsv_{k+\ell} \\ 0 \le \ell < M-k}}
  \EE\left[\left.
  \sum_{0 \le \ell < M-k}
  \left[ \bsv_{k+\ell}^T \Xi(t_{k+\ell}) \bsv_{k+\ell} + \lambda \bsx_{k+\ell}^T \Sigma(t_{k+\ell}) \bsx_{k+\ell} \right] \Delta t
  \right|
  \bsx_k, \Xi(t_k), \Sigma(t_k)
  \right]
  \\[3mm]
  &\qquad=
  \min_{\bsv_k}
  \Bigl\{
  \left[ \bsv_k^T \Xi(t_k) \bsv_k + \lambda \bsx_k^T \Sigma(t_k) \bsx_k \right] \Delta t
  \\
  &\qquad\quad+
  \min_{\substack{\bsv_{k+\ell} \\ 1 \le \ell < M-k}}
  \EE\left[\left.
  \sum_{1 \le \ell < M-k}
  \left[ \bsv_{k+\ell}^T \Xi(t_{k+\ell}) \bsv_{k+\ell} + \lambda \bsx_{k+\ell}^T \Sigma(t_{k+\ell}) \bsx_{k+\ell} \right] \Delta t
  \right|
  \bsx_{k+1}, \Xi(t_k), \Sigma(t_k)
  \right]
  \Bigr\}
  .
\end{align*}
At this point we want to remark that the part inside the expectation depends on $\bsv_k$ through $\bsx_{k+1}$.
The idea is now to use Monte Carlo for this expectation by generating instances of $\Xi(t_{k+\ell})$ and $\Sigma(t_{k+\ell})$, for $1 \le \ell < M-k$, and then using the RHS scheme with these sampled values.
We can then write the cost at a future step $k+\ell$ using\RefEq{eq:xkl} as
\begin{align*}
  \bsv_{k+\ell}^T \Xi(t_{k+\ell}) \bsv_{k+\ell} + \lambda \bsx_{k+\ell}^T \Sigma(t_{k+\ell}) \bsx_{k+\ell}
  &=
  \bsx_{k+\ell}^T 
  \left[ \frac{\Omega_{k+\ell}^T-I_n}{\Delta t} \; \Xi(t_{k+\ell}) \; \frac{\Omega_{k+\ell}-I_n}{\Delta t} + \lambda \, \Sigma(t_{k+\ell}) \right] 
  \bsx_{k+\ell}
  \\
  &=
  \bsx_{k+1}^T
  \; 
  Q_{k,\ell}(\Xi(\cdot), \Sigma(\cdot))
  \;
  \bsx_{k+1}
  ,
\end{align*}
where
\begin{align}
	Q_{k,\ell}(\Xi(\cdot), \Sigma(\cdot))
	&:=
	V^T(k,\ell,\Xi(\cdot),\Sigma(\cdot)) \;
  \left[ \frac{\Omega_{k+\ell}^T-I_n}{\Delta t} \; \Xi(t_{k+\ell}) \; \frac{\Omega_{k+\ell}-I_n}{\Delta t} + \lambda \, \Sigma(t_{k+\ell}) \right] 
  \; V(k,\ell,\Xi(\cdot),\Sigma(\cdot))
    . 
    \label{eq:Q_rhmc1}
\end{align}
The total future trading cost is therefore given by
\begin{align*}
  \sum_{1 \le \ell < M-k}
  \left[ \bsv_{k+\ell}^T \Xi(t_{k+\ell}) \bsv_{k+\ell} + \lambda \bsx_{k+\ell}^T \Sigma(t_{k+\ell}) \bsx_{k+\ell} \right]
  =
  \bsx_{k+1}^T\;
  A_{k+1}(\Xi(\cdot), \Sigma(\cdot))
  \;
  \bsx_{k+1}
\end{align*}
where 
\begin{align*}
	A_{k+1}(\Xi(\cdot), \Sigma(\cdot))
	&:=
	\sum_{1 \le \ell < M-k} Q_{k,\ell}(\Xi(\cdot), \Sigma(\cdot))
	.
\end{align*}
Finally, the expected future trading cost is given by
\begin{multline}
  \EE\left[\left.
  \sum_{1 \le \ell < T-k}
  \left[ \bsv_{k+\ell}^T \Xi(t_{k+\ell}) \bsv_{k+\ell} + \lambda \bsx_{k+\ell}^T \Sigma(t_{k+\ell}) \bsx_{k+\ell} \right]
  \right|
  \bsx_{k+1}, \Xi(t_k), \Sigma(t_k)
  \right]
  \\=
  \bsx_{k+1}^T \; \overline{A}_{k+1}(\Xi(t_k), \Sigma(t_k)) \; \bsx_{k+1}
  ,
  \label{eq:futurecost}
\end{multline}
where we defined
\begin{align}
	\overline{A}_{k+1}(\Xi(t_k), \Sigma(t_k))
	&:=
	\EE\left[\left.A_{k+1}(\Xi^{\omega}(\cdot), \Sigma^{\omega}(\cdot)) \right| \Xi(t_k), \Sigma(t_k) \right]
	,
	\label{eq:overlineA}
\end{align}
to be the element wise expectation of the matrix $A_{k+1}(\Xi^{\omega}(\cdot), \Sigma^{\omega}(\cdot))$.

Calculating the matrix $A_{k+1}(\Xi(\cdot), \Sigma(\cdot))$ for a specific instance of $\Xi(t_{k+\ell})$ and $\Sigma(t_{k+\ell})$, for $1 \le \ell < M-k$, costs $O((M-k) \, 5 n^3)$ from matrix products, making use of\RefEq{eq:Vrecur}.
If we assume the cost of calculating an $\Omega_k$ matrix is also $O(n^3)$, e.g., by making use of an eigen decomposition, then the cost per instance of $A_{k+1}(\Xi(\cdot), \Sigma(\cdot))$ is $O((M-k) \, n^3)$.
To approximate this expectation we use a Monte Carlo (or quasi-Monte Carlo) method with $N$ samples such that approximating the expected value $\overline{A}_{k+1}(\Xi(t_k), \Sigma(t_k))$ costs $O(N \, (M-k) \, n^3)$.
(We ignore the cost of generating the matrices $\Xi^{(i)}(t_{k+\ell})$ and $\Sigma^{(i)}(t_{k+\ell})$, for $1 \le \ell < M-k$, $1 \le i \le N$ as we assume this to be quadratic in $n$, i.e., of order $O(N \, (M-k) \, n^2)$.)
It will turn out that finding the minimum value will cost $O(n^3)$, see \RefProp{prop:optimalv}.
Thus the traders' cost in step $k$ is $O(N \, (M-k) \, n^3)$.

We can now continue with the minimization of $\bsv_k$ having removed the minimization problem for $\bsv_{k+\ell}$, $1 \le \ell < M-k$, by using the RHS scheme ``on average'' and using the expected continuation cost of the RHS method in terms of $\bsx_{k+1}$ (and thus $\bsv_k$).
In fact, it is possible to minimize directly over $\bsx_{k+1}$, omitting the need to calculate $\bsv_k$.
For the RHMC-I scheme we now look at
\begin{multline*}
  \min_{\bsv_k} \Bigl\{
    \left[ \bsv_k^T \Xi(t_k) \bsv_k + \lambda \bsx_k^T \Sigma(t_k) \bsx_k \right] \Delta t
    +
    \bsx_{k+1}^T \; \overline{A}_{k+1}(\Xi(t_k), \Sigma(t_k)) \; \bsx_{k+1} \Delta t
  \Bigr\}
  \\
  =
  \bsx_k^T \left(\Xi(t_k)\frac{1}{\Delta t^2}+\lambda\Sigma(t_k)\Delta t\right) \bsx_k  
  \\+ 
  \min_{\bsx_{k+1}} \Bigl\{
    \frac{1}{\Delta t} \left(\bsx_{k+1}^T \;\Xi(t_k) \; \bsx_{k+1}-\bsx_{k+1}^T \;\Xi(t_k) \; \bsx_k - \bsx_k^T \;\Xi(t_k) \; \bsx_{k+1}\right)
    \\
    + \bsx_{k+1}^T \; \overline{A}_{k+1}(\Xi(t_k), \Sigma(t_k)) \; \bsx_{k+1} \Delta t
  \Bigr\}
  .
\end{multline*}
The above is a quadratic polynomial in $\bsx_{k+1}$.

\begin{proposition}
	In each time step $t_k$ of the discretization, given the sampled matrix $\overline{A}_{k+1}$, there exists a unique optimal position $\bsx_{k+1}$, which is given as the solution to
	\begin{align}
	  \left( 
	  \Xi(t_k)
	  + \overline{A}_{k+1}  \, (\Delta t)^2
	  \right) 
	  \bsx_{k+1}
	  &=
	  \Xi(t_k)\, 
	  \bsx_k
	  .
	  \label{eq:optControl}
	\end{align}
	\label{prop:optimalv}
\end{proposition}
\begin{proof}
	We reduce the notational overload in the minimization expression by setting $\bsx=\bsx_{k+1}$, $\Xi = \Xi(t_k)$, $\overline{A} = \overline{A}_{k+1}(\Xi(t_k), \Sigma(t_k))$.
	We need to minimize the function
	\begin{align*}
	  f(\bsx)
	  &=
	  \bsx^T \left(\Xi\Delta t+\frac{\overline{A}}{\Delta t}\right) \bsx 
	  - (\bsx^T\, \Xi\, \bsx_k + \bsx_k^T\, \Xi\, \bsx) \frac{1}{\Delta t}
	  .
	\end{align*}
	Setting the gradient with respect to $\bsx$ to zero we obtain
	\begin{align*}
	  \nabla f(\bsx)
	  &=
	  2\left(\Xi\Delta t+\frac{\overline{A}}{\Delta t}\right) \bsx  
	  - \frac{2\Xi}{\Delta t}\,\bsx_k 
	  =
	  \bszero
	  ,
	\end{align*}
	and thus $\bsx$ is the solution to
	\begin{align*}
	  \left( 
	  \Xi
	  + \overline{A}  \, (\Delta t)^2
	  \right) 
	  \bsx
	  &=
	  \Xi\, 
	  \bsx_k
	  ,
	\end{align*}
	provided there exists a unique minimum.
	This can be checked by the positive definiteness of the Hessian matrix which is given by
	\begin{align*}
	  H(f)
	  &=
	  2\left(\Xi \, \Delta t
	  + \frac{\overline{A}}{\Delta t}\right)
	  .
	\end{align*}
	Remember that $\Xi$ is positive definite by assumption.
	We next show that $\overline{A}=\overline{A}_{k+1}$ is positive definite.

	We have that $\overline{A}_{k+1}$ is the sum of matrices
	\begin{align*}
	  Q_{k,\ell}(\Xi(\cdot), \Sigma(\cdot))
	  &=
	  V^T(k,\ell,\Xi(\cdot),\Sigma(\cdot)) \;
	  \left[\frac{\Omega_{k+\ell}^T-I_n}{\Delta t} \; \Xi(t_{k+\ell}) \; \frac{\Omega_{k+\ell}-I_n}{\Delta t} + \lambda \, \Sigma(t_{k+\ell}) \right] 
	  \; V(k,\ell,\Xi(\cdot),\Sigma(\cdot))
	  .
	\end{align*}
	For $\ell=1$ we have $V(k,\ell,\Xi(\cdot),\Sigma(\cdot)) = I_n$ and thus
	\begin{align*}
	  Q_{k,1}(\Xi(\cdot), \Sigma(\cdot))
	  &=
	  \frac{\Omega_{k+1}^T-I_n}{\Delta t} \; \Xi(t_{k+1}) \; \frac{\Omega_{k+1}-I_n}{\Delta t} + \lambda \, \Sigma(t_{k+1})
	  .
	\end{align*}
	By \RefProp{prop:stability} we know that the matrices $\Omega_{k+1}$ have eigenvalues in the interval $(0,1)$.
	Therefore, 	$\Omega_{k+1}-I_n$ has eigenvalues in the interval $(-1,0)$.
	Multiplying right and left makes the product $(\Omega_{k+1}^T-I_n) \; \Xi(t_{k+1}) \; (\Omega_{k+1}-I_n)$ positive definite, since $\Xi(t_{k+1})$ is positive definite by construction. 
	For $\ell >1$ we keep on adding $\Omega_{k+\ell-m}$ to the right and to the left.
	Therefore $\overline{A}_{k+1}$ is positive definite.
\end{proof}

\begin{proposition}
		The RHMC-I algorithm is stable, i.e., the eigenvalues of 
		\begin{align}
	  		&
	  		\left( 
	  		\Xi(t_k)
	  		+ \overline{A}_{k+1}  \, (\Delta t)^2
	  		\right)^{-1}
	  		\Xi(t_k)
	  		\label{eq:stableRHMC}
		\end{align}
	are positive and bounded by one for all $k$.
	\label{prop:stableRHMC}
\end{proposition}
\begin{proof}
	Denote by $\mu_1,\ldots,\mu_n$ the eigenvalues of\RefEq{eq:stableRHMC} in increasing order. 
	Because the product of two positive definite matrices has positive eigenvalues (see \citet{HJ1985}), we have $\mu_1>0$.

	For the upper bound we find
	\begin{align*}
		\mu_n\left(\left(\Xi(t_k)+\overline{A}_{k+1} \, (\Delta t)^2\right)^{-1}\Xi(t_k)\right)
		&\leq
		\mu_n\left(\left(\Xi(t_k)+\overline{A}_{k+1} \, (\Delta t)^2\right)^{-1}\right)\mu_n\left(\Xi(t_k)\right) \\
		&=
		\left(\mu_1\left(\Xi(t_k)+\overline{A}_{k+1} \, (\Delta t)^2\right)\right)^{-1}\mu_n\left(\Xi(t_k)\right).
	\end{align*}
	From the Weyl inequalities \citet{bhatia2001linear} we have that
	\begin{align*}
		\left(\mu_1\left(\Xi(t_k)+\overline{A}_{k+1} \, (\Delta t)^2\right)\right)^{-1}
		&\leq
		\left(\mu_n(\Xi(t_k))+(\Delta t)^2 \mu_1(\overline{A}_{k+1})\right)^{-1},
	\end{align*}
	leading to the inequality
	\begin{align*}
		\mu_n\left(\left(\Xi(t_k)+\overline{A}_{k+1} \, (\Delta t)^2\right)^{-1}\Xi(t_k)\right)
		&\leq
		\frac{\mu_n\left(\Xi(t_k)\right)}{\mu_n(\Xi(t_k))+(\Delta t)^2 \mu_1(\overline{A}_{k+1})}
		\leq 1
	\end{align*}	
	where the last inequality follows from the positive definiteness of $\overline{A}_{k+1}$, which was proven in \RefProp{prop:optimalv}.
\end{proof}

\subsection*{RHMC-II: Rolling horizon Monte Carlo with CC till the end}

Under the RHMC-I method the trading rate in step $t_k$ was decided by using the RHS scheme and calculating the matrix $A_{k+1}(\Xi(\cdot),\Sigma(\cdot))$ , and then using \RefProp{prop:optimalv}. 
It is of course not mandatory to use RHS as the suboptimal control.
We could just as well use the CC to determine $A_{k+1}(\Xi(\cdot),\Sigma(\cdot))$ which will reduce the computational complexity.
The difference with RHMC-I therefore is that we now assume the matrices $\Xi(s)$ and $\Sigma(s)$ to be fixed from time $t_{k+1}$ on to calculate the suboptimal control.
We can therefore use\RefEq{eq:CC} to write\RefEq{eq:futurecost}, for $1\leq \ell< M-k$, as
\begin{align*}
  \bsv_{k+\ell}^T \Xi(t_{k+\ell}) \bsv_{k+\ell} + \lambda \bsx_{k+\ell}^T \Sigma(t_{k+\ell}) \bsx_{k+\ell}
  &=
  \bsx_{k+1}^T
  \; 
  Q_{k,\ell}(\Sigma(\cdot),\Xi(\cdot))
  \;
  \bsx_{k+1}
  ,
\end{align*}
where $Q_{k,\ell}$ is now defined as
\begin{align}
  Q_{k,\ell}(\Sigma(\cdot),\Xi(\cdot))
  &=
 \Omega_{k,\ell}'^T \; \Xi(t_{k+\ell}) \;  \Omega_{k,\ell}' + \lambda \Omega_{k,\ell}^T \; \Sigma(t_{k+\ell}) \;  \Omega_{k,\ell}
 \label{eq:Q_rhmc2}
\end{align}
and $\Omega_{k,\ell}=\Omega(t_{k+\ell}-t_{k+1},T-t_{k+1},\Xi(t_{k+1}),\Sigma(t_{k+1}))$.
Both \RefProp{prop:optimalv} and \RefProp{prop:stableRHMC} are still valid for the RHMC-II method.

Assuming the left rectangle rule we arrive at a cost $O(N (M-k) n^3)$ which is the same complexity as RHMC-I but it will in practice be (much) faster since the matrix $\Omega$ can now be calculated as a function of time instead of being updated on each of the future trading dates and there is no more need to compute the function $V$ which is essentially a cumulative matrix product. 

\bigskip

An overview of both of the algorithms is given in \RefAlg{alg:MCRHalg}.
Note that in the single asset case under coordinated variation and for $\lambda\bar{\sigma}^2/\bar{\eta}\rightarrow \infty$ our RHMC-I scheme converges to the optimal solution since \RefProp{prop:RHS} shows that the RHS converges to the optimal solution, and the continuation values are calculated using this optimal control. 
Other than this case we expect RHMC-I to improve significantly on the rolling horizon solution for any choice of market parameters.
These outlooks can be justified since the RHS only considers the current market conditions and gives the trading rate under the assumption that these conditions remain constant for the remaining trading period.
Our algorithm, on the other hand, uses a Monte Carlo procedure to gain information about future trading, and chooses a trading rate accordingly.

It is at this point not clear whether the RHMC-II method will provide an improvement over the RHS.
At least in the case $n=1$ under coordinated variation and $\lambda\bar{\sigma}^2/\bar{\eta}\rightarrow\infty$ we expect the RHMC-II method to be worse than the RHS, since the RHS will tend to be optimal whereas the CC will not.
The numerical results in \RefSec{sec:NumResults} show that the RHMC-II algorithm indeed underperforms in this case, but as soon as coordinated variation is not assumed, it performs similarly to the RHMC-I algorithm.
Interestingly, in the case $n=2$ it appears that there is almost no difference between the RHMC-I and RHMC-II algorithms for all market parameters considered.

\begin{algorithm}
 \caption{The Monte Carlo rolling horizon method for optimal asset execution.}
 \label{alg:MCRHalg}
 \begin{algorithmic}
  \STATE
  \STATE Fix a set of $M+1$ trading times, equally spaced with length $\Delta t=T/M$.
  \STATE
  \STATE At time $t_k$ observe $\Xi(t_k)$ and $\Sigma(t_k)$.
  \STATE Sample $N$ paths for $\Xi(\cdot)$ and $\Sigma(\cdot)$ given $\Sigma(t_k)$ and $\Xi(t_k)$ using (quasi-)Monte Carlo.
  \STATE Calculate the matrix $\overline{A}$\RefEq{eq:overlineA} by using\RefEq{eq:Q_rhmc1} for the RHMC-I method or\RefEq{eq:Q_rhmc2} for the RHMC-II method.
  \STATE Calculate $\bsx_{k+1}$ using\RefEq{eq:optControl}.
  \STATE Trade $\bsx_{k+1}-\bsx_k$ assets.
  \STATE
 \end{algorithmic}
\end{algorithm}

\subsection{An a posteriori discrete optimal solution}\label{sec:OPT}

It is possible to compute the optimal control for the dynamic problem given a discretization and assuming the paths of the market parameters are known.
This is of course of no use for a trader facing an execution problem, but it allows us to check how the cost of trading is situated for all methods against the cost when trading with the optimal discrete control.

Assume for now that there is only trading in one asset, i.e., $n=1$.
We use again the discretization of time $(t_k)_{k=0,\ldots,M}$ with step length $\Delta t=T/M$, and approximate the integrals using a left-point rule.
Given the paths of the market parameters, $\Xi(t_k)$ and $\Sigma(t_k)$, $k=0,\ldots,M$, the optimal control problem corresponds to solving the minimization problem 
\begin{equation*}
	\begin{aligned}
		& \underset{v_k,k=0,\ldots,M}{\text{minimize}}
		& & \Delta t\sum_{k=0}^{M-1} v_k^2 \Xi(t_k)  + \lambda x_k^2\Sigma(t_k) \\
		& \text{subject to}
		& & x_0=X, 
	\end{aligned}
\end{equation*}
By noting that
\begin{align*}
	x(t)
	&=
	-\int_t^T v(s)ds
	,
\end{align*}
we find for the discretized points $x(t_k)$ 
\begin{align*}
	x_k
	&=
	-\sum_{m=k}^{M-1}v_m\Delta t.
\end{align*}
Therefore, the minimization problem can be written as
\begin{equation*}
	\begin{aligned}
		& \underset{v_k,k=0,\ldots,M}{\text{minimize}}
		& & \Delta t\sum_{k=0}^{M-1} v_k^2 \Xi(t_k) + \lambda\Delta t^2 \Sigma(t_k) \left( \sum_{m=k}^{M-1} v_m \right)^2 \\
		& \text{subject to}
		& & -\sum_{m=0}^{M-1}v_m\Delta t=X, 
	\end{aligned}
\end{equation*}
Some elementary calculations show that the above problem is equivalent to 
\begin{equation*}
	\begin{aligned}
		& \underset{v_k,k=0,\ldots,M}{\text{minimize}}
		& & \Delta t\sum_{k=0}^{M-1} v_k^2 \Xi(t_k) + \lambda\Delta t^3 \sum_{k=0}^{M-1}\sum_{\ell=0}^{M-1}v_kv_{\ell} \sum_{m=0}^{\min(k,\ell)} \Sigma(t_m) \\
		& \text{subject to}
		& & -\sum_{m=0}^{M-1}v_m\Delta t=X, 
	\end{aligned}
\end{equation*}
By introducing the matrix $\tilde{\Sigma}$ where $\tilde{\Sigma}_{k\ell}=\sum_{m=0}^{\min(k,\ell)} \Sigma(t_m)$, the diagonal matrix $\tilde{\Xi}$ where $\tilde{\Xi}_{kk}=\Xi(t_k)$ and the vector $\bsv^t$ which is the time-discretized control $v$, i.e., $v^t_m=v_m$, we can rewrite the above problem to
\begin{equation*}
	\begin{aligned}
		& \underset{\bsv^t}{\text{minimize}}
		& & (\bsv^t)^T \left(\Delta t\tilde{\Xi}+\lambda\Delta t^3\tilde{\Sigma}\right)\bsv^t \\
		& \text{subject to}
		& & -\Delta t \bsone_M^T \bsv^t=X, 
	\end{aligned}
\end{equation*}
Since the matrix $\Delta t\tilde{\Xi}+\lambda\Delta t^3\tilde{\Sigma}$ is positive definite, a unique minimizer exists.

The more general case $n>1$ follows easily from this construction.
Denote by $\bsv_k^t$ the time-discretized vectors of asset $k$, $k=1,\ldots,n$
The minimization problem becomes
\begin{equation*}
	\begin{aligned}
		& \underset{\bsv^t_{1,\ldots,n}}{\text{minimize}}
		& & (\bsv_{1,\ldots,n}^t)^T \left(\Delta t\begin{pmatrix}\tilde{\Xi}^{(11)} & \cdots & \tilde{\Xi}^{(1n)} \\ \vdots & \ddots & \vdots \\ \tilde{\Xi}^{(n1)} & \cdots & \tilde{\Xi}^{(nn)} \end{pmatrix} + \lambda \Delta t^3 \begin{pmatrix}\tilde{\Sigma}^{(11)} & \cdots & \tilde{\Sigma}^{(1n)} \\ \vdots & \ddots & \vdots \\ \tilde{\Sigma}^{(n1)} & \cdots & \tilde{\Sigma}^{(nn)} \end{pmatrix}\right) \bsv_{1,\ldots,n}^t \\
		& \text{subject to}
		& & -\Delta t\bsw_1^T\bsv_{1,\ldots,n}^t=X_1, \\
		& & &\hspace{2.5cm} \vdots \\
		& & & -\Delta t\bsw_n^T\bsv_{1,\ldots,n}^t=X_n, 
	\end{aligned}
\end{equation*}
where $\tilde{\Xi}^{(ij)}$ is the diagonal matrix with $\tilde{\Xi}^{(ij)}_{kk}=\Xi_{ij}(t_k)$, $\tilde{\Sigma}^{(ij)}_{k\ell}=\sum_{m=0}^{\min(k,\ell)} \Sigma_{ij}(t_m)$,  
\begin{align*}
	\bsv_{1,\ldots,n}^t
	&=
	\begin{pmatrix} \bsv_1^t \\ \vdots \\ \bsv_n^t \end{pmatrix}
	&\text{and}&&
	\bsw_k
	=
	\begin{pmatrix} \bszero_{M(k-1)} \\ \bsone_M \\ \bszero_{M(n-k-2)} \end{pmatrix}.
\end{align*}

\section{Numerical results}\label{sec:NumResults}

In this section we numerically illustrate RHMC-I and RHMC-II, and compare it to the CC, RHS and discrete optimal solutions.
The aim of our method is to outperform the RHS solution.
The results in this section show that we are successful at this goal.

To choose the number of simulations used for calculating the expected continuation value, we have looked at the cost of trading for different parameter choices in function of the number of samples used.
A typical example is given in \RefFig{fig:1A_conv} for $n=1$ and \RefFig{fig:2A_conv} for $n=2$, using the RHMC-I method.
It is clear that using quasi-Monte Carlo (the red curves) is advantageous over using plain Monte Carlo (the blue curves).
For our methods we use a Sobol' sequence with parameters from \citet{JK2008}.
We will choose $N=500$ for our numerical experiments, but we would also like to point out that taking $N=200$ using quasi-Monte Carlo seems sufficient for practical applications.

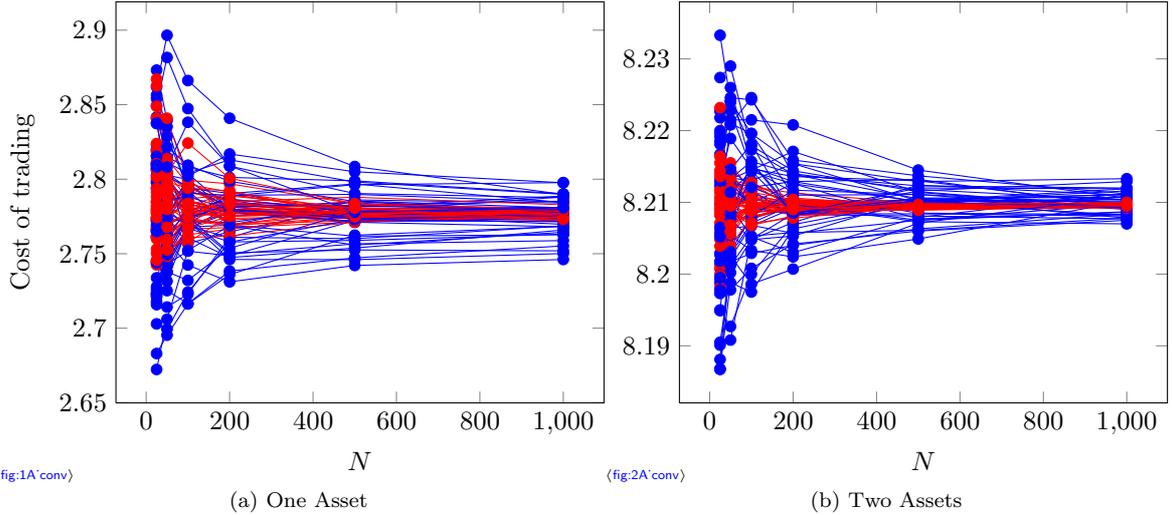
\begin{figure}[h]
\centering
\subfloat[One Asset]{\label{fig:1A_conv}
\begin{tikzpicture}
        
          \begin{axis}
          [ width=8cm,
          	xlabel=$N$,
            ylabel={Cost of trading},
            ]
          \foreach \yindex in {1,...,41}
            {   
            \addplot[mark=*,blue] table [y index = \yindex] {masterMC1.txt};
            \addplot[mark=*,red] table [y index = \yindex] {masterQMC1.txt};
            } 
          \end{axis}
          
        \end{tikzpicture}}
\subfloat[Two Assets]{\label{fig:2A_conv}
\begin{tikzpicture}
        
          \begin{axis}
          [ width=8cm,
          	xlabel=$N$,
            ]
          \foreach \yindex in {1,...,41}
            {   
            \addplot[mark=*,blue] table [y index = \yindex] {masterMC2.txt};
            \addplot[mark=*,red] table [y index = \yindex] {masterQMC2.txt};
            } 
          \end{axis}
          
        \end{tikzpicture}}
\caption{The cost of trading given specific paths for the liquidity and volatility parameters, in function of $N$, the number of (quasi-)Monte Carlo paths used to determine the mean continuation value. Each blue curve corresponds to a different seed for the Monte Carlo numbers used, and each red curve corresponds to a different digital shift used to obtain the quasi-Monte Carlo sample.}\label{fig:1A2A_conv}
\end{figure}

\subsection{One asset under coordinated variation}

For the one asset case, assuming coordinated variation holds, we ran experiments using fixed parameters $T=10$, $\Delta t=1/100$, $\beta_1=\delta_1=1$, and different values of $\bar{\sigma}$, $\bar{\eta}$ and $\lambda$.
The dimension of the problem, at some arbitrary time step $k\in\{0,1,\ldots,1000\}$ is $k$, since there is only one stochastic process driving the market parameters.

\RefTab{table:ex1ACV} gives an overview of the cost of trading for different parameter values using $200$ simulation runs, i.e., $200$ sample paths of $\xi$ on which we used the different execution algorithms.
Since this is also the setting considered in \citet{almgren2012optimal}, we have implemented the finite difference scheme as outlined in said paper.
We will call this solution the optimal continuous solution.
It should be noted that the discrete optimal solution and the continuous optimal solution do not need to coincide.

The effect of \RefProp{prop:RHS2} is clearly visible in the table: the RHS and CC have similar costs of trading when $\lambda\bar{\sigma}^2/\bar{\eta}$ is smaller than $10^{-3}$.
It is also in these cases that both RHMC-I and RHMC-II outperform the RHS method significantly, reducing the extra cost over the continuous optimal solution to about a fourth of that of RHS.
This reduction in cost is constant over all parameter choices for which $\lambda\bar{\sigma}^2/\bar{\eta}$ is smaller than $10^{-3}$. 
Both our methods have similar, if not identical, reductions in cost.

On the other hand, for $\lambda\bar{\sigma}^2/\bar{\eta}$ larger than $10^{-3}$ the effect of \RefProp{prop:RHS} becomes visible in the results, with the extra cost of the RHS compared to the continuous optimal solution dropping significantly. 
Our RHMC-I method keeps outperforming the RHS method, especially for the cases where $\lambda\bar{\sigma}^2/\bar{\eta}$ is of the order $10^{-2}$.
The RHMC-II method starts to lag behind the RHMC-I method for these cases, with smaller reductions in cost.
For the cases where $\lambda\bar{\sigma}^2/\bar{\eta}$ is of the order $10^{-1}$ the advantage of the RHMC-I method is greatly reduced, because the extra cost of the RHS method is almost zero.
The RHMC-II method even fails to outperform the RHS method in these cases.
This should not come as a surprise, since the RHMC-I method uses the RHS as suboptimal control, so if the RHS converges to the optimal control, the RHMC-I method also converges to the optimal control by construction.
Because RHMC-II uses the CC as suboptimal control, it does not show this convergence.

It is noteworthy to mention the increase in cost should a trader use the CC solution when market parameters are not constant.

\subsection{One asset}

We ran numerical experiments using fixed parameters $T=10$ and $\beta_1=\beta_2=\delta_1=\delta_2=1$, for different values of $\bar{\sigma}$, $\bar{\eta}$, $\lambda$ and $\varrho$, and a time discretization $\Delta t=1/100$.
The dimension of the problem, at some arbitrary time step $k\in\{0,1,\ldots,1000\}$ is $2k$, since there are two stochastic processes driving the market parameters.

\RefTab{table:ex1A} gives an overview of the cost of trading for different parameter values using $200$ simulation runs, i.e., $200$ sample paths of $\Xi$ and $\Sigma$ on which we used the different execution algorithms.
The percentages indicate the increase in cost compared to the discrete optimal solution.
The effects of \RefProp{prop:RHS2} are again clearly visible in the table: for values of $\lambda\bar{\sigma}^2/\bar{\eta}$ of the order $10^{-3}$ or smaller the RHS and CC algorithms almost coincide.
Our RHMC-I and RHMC-II method significantly outperform the RHS method, cutting the extra cost over the discrete optimal solution to less than half. 
Both methods have almost identical performance.

Interestingly, even for $\lambda\bar{\sigma}^2/\bar{\eta}$ of order larger than $10^{-3}$ both RHMC-I and RHMC-II methods significantly outperform the RHS method, as opposed to the coordinated variation case.
The reduction for the RHMC-I method is still around half, whereas the RHMC-II method is now showing reductions that are less than half of the extra cost.

\subsection{Two assets}\label{sec:2A}

We ran numerical experiments using fixed parameters $T=10$, $\bar{\sigma}_1^2=1/500$, $\bar{\sigma}_1^2=3/1000$, $\bar{\eta}_{11}=1/400$, $\bar{\eta}_{12}=\bar{\eta}_{21}=1/1000$ and $\beta_k=\delta_k=1$ for $k=1,\ldots,5$.
The correlation matrix of the Brownian motions driving the market parameters is fixed as
\begin{align*}
 \varrho
 &=
 \frac{1}{10}\begin{bmatrix}
    10 & 8 & 1 & -6 & -6 \\
    8 & 10 & 1 & -6 & -6 \\
    1 & 1 & 10 & -1 & -1 \\
    -6 & -6 & -1 & 10 & 7 \\
    -6 & -6 & -1 & 7 & 10
 \end{bmatrix}.
\end{align*}
Experiments were run with different values of $\bsx_0$, $\bar{\eta}_{22}$, $\rho$ and $\lambda$ .
Time is discretized with step length $\Delta t=1/100$.
The dimension of the problem, at some arbitrary time step $k\in\{0,1,\ldots,1000\}$ is $5k$, since there are five stochastic processes driving the market parameters.

\bigskip

\RefTab{table:ex2A} shows the costs of all methods and the improvements over the RHS of our method.
It also shows the largest absolute value in the matrix $\lambda\Xi^{-1}\Sigma$, an indication of how close this matrix is to the zero matrix.
A first glance at the results shows the interesting observation that there seems to be almost no difference between the RHMC-I and RHMC-II algorithm for all cases considered.
The results are similar to the one asset case when $\lambda\Xi^{-1}\Sigma$ is fairly close to zero, i.e., the extra cost of trading over using the optimal control for the RHMC-I and RHMC-II methods are about one third of those of RHS for $\lambda\Xi^{-1}\Sigma$ smaller than the order $10^{-6}$.
For $\lambda\Xi^{-1}\Sigma$ of order $10^{-4}$ the RHMC-I and RHMC-II methods still perform similarly, still outperforming the RHS significantly in all the cases considered.
This is especially true in the case where opposite initial positions have to be traded: under these circumstances, costs are cut by a fifth to a sixth compared to RHS for both methods.
As mentioned before, there seems to be evidence though that there is some extension to \RefProp{prop:RHS} to multiple assets, but there is also evidence that it is not sufficient to let $\lambda\Xi^{-1}\Sigma\rightarrow 0_n$.
Also note the huge increase in cost in some cases when using CC, which can be up to three times the cost compared to the optimal solution.

The results for two assets are again very satisfactory (even more than the one asset case).

\section{Conclusion and outlook}\label{sec:end}

In this paper we have studied the static solution of the multidimensional extension to the model in \citet{almgren2012optimal}.
We also presented a new result on the rolling horizon strategy (RHS), an approximate scheme constructed in the same paper to trade in a dynamic market.
Our aim was to develop a method that performs better than the RHS, but is still easier to compute than the optimal solution.
This lead to the rolling horizon Monte Carlo methods (RHMC), that determine the optimal trading rate in each time step by balancing the cost of trading over the current time interval against the projected future trading costs.
The future trading costs are determined using (quasi-)Monte Carlo and a suboptimal trading strategy, in our case the RHS (used in RHMC-I) and the constant coefficient solution (used in RHMC-II).
It was argued to perform better than the RHS since the latter only considers the current market information, whereas our method uses current information as well as a prediction on future costs.
The advantages of this scheme are that it is easy to understand and there is no need to solve a high-dimensional PDE (when using for instance a finite-difference solver) or regressions (when using least-squares Monte Carlo).
It turns out that using our method we can reduce the extra cost of trading over that of the RHS to as much as one sixth.
Furthermore, our methods seem to be more consistent in cutting trading cost compared to the RHS when there is more than one asset to be traded.

It would be very interesting to see if these results carry over to other models as well, especially when price impact is not instantaneous but temporary, see for instance \citet{alfonsi2009order,GSA2012}.
Another interesting expansion would be to include directional bets in the strategy, see \citet{AL2006} and \citet{EF2007}. Checking what happens if the small-impact approximation is violated, i.e., if the variance in the cost of trading also depends on the liquidity and chosen trading rates could also lead to interesting results.

\section*{Appendix: Tables}


\begin{landscape}
\begin{table}[h]\footnotesize
\centering
\begin{tabular}{ l||c| c c|c c|c c|c c|c|c}
  &  Continuous   & \multicolumn{2}{c}{ } & \multicolumn{2}{c}{ } & \multicolumn{2}{c}{ } & \multicolumn{2}{c}{ } & Discrete &  \\
    $(\bar{\sigma},\bar{\eta},\lambda)$ &  Optimal  & \multicolumn{2}{c}{CC} & \multicolumn{2}{c}{RHS} & \multicolumn{2}{c}{RHMC-I} & \multicolumn{2}{c}{RHMC-II} &  Optimal & $\lambda\sigma^2/\eta$\\
\hline 
$(0.031$,$0.002$,$10^{-5}$) &$1.97 $& $2.58 $ & ($31$\%) & $2.58 $ & ($31$\%) & $2.11 $ & ($7.2$\%) & $2.11 $ & ($7.2$\%) & $1.82 $ & $5\times 10^{-6}$ \\
$(0.031$,$0.002$,$0.001$) &$2.00 $& $2.62 $ & ($30$\%) & $2.60 $ & ($29$\%) & $2.14 $ & ($6.8$\%) & $2.14 $ & ($6.9$\%) & $1.86 $ & $5\times 10^{-4}$ \\
$(0.031$,$0.002$,$0.1$) &$4.60 $& $5.43 $ & ($17$\%) & $4.73 $ & ($2.7$\%) & $4.62 $ & ($0.39$\%) & $4.69 $ & ($2.0$\%) & $4.56 $ & $5\times 10^{-2}$ \\
$(0.031$,$0.003$,$10^{-5}$) &$2.96 $& $3.87 $ & ($31$\%) & $3.87 $ & ($31$\%) & $3.17 $ & ($7.2$\%) & $3.17 $ & ($7.2$\%) & $2.73 $ & $3.3\times 10^{-6}$ \\
$(0.031$,$0.003$,$0.001$) &$2.99 $& $3.91 $ & ($30$\%) & $3.89 $ & ($30$\%) & $3.20 $ & ($6.9$\%) & $3.20 $ & ($7.0$\%) & $2.77 $ & $3.3\times 10^{-4}$ \\
$(0.031$,$0.003$,$0.1$) &$5.81 $& $6.92 $ & ($19$\%) & $6.08 $ & ($4.5$\%) & $5.85 $ & ($0.64$\%) & $5.94 $ & ($2.2$\%) & $5.71 $ & $3.3\times 10^{-2}$ \\
$(0.063$,$0.002$,$10^{-5}$) &$1.97 $& $2.58 $ & ($31$\%) & $2.58 $ & ($30$\%) & $2.11 $ & ($7.2$\%) & $2.11 $ & ($7.2$\%) & $1.82 $ & $2\times 10^{-5}$ \\
$(0.063$,$0.002$,$0.001$) &$2.11 $& $2.72 $ & ($29$\%) & $2.65 $ & ($25$\%) & $2.23 $ & ($5.7$\%) & $2.24 $ & ($6.1$\%) & $1.97 $ & $2\times 10^{-3}$ \\
$(0.063$,$0.002$,$0.1$) &$8.99 $& $10.4 $ & ($15$\%) & $9.05 $ & ($0.75$\%) & $9.04 $ & ($0.56$\%) & $9.14 $ & ($1.7$\%) & $9.02 $ & $2\times 10^{-1}$ \\
$(0.063$,$0.003$,$10^{-5}$) &$2.96 $& $3.88 $ & ($31$\%) & $3.87 $ & ($30$\%) & $3.17 $ & ($7.2$\%) & $3.17 $ & ($7.2$\%) & $2.74 $ & $1.3\times 10^{-5}$ \\
$(0.063$,$0.003$,$0.001$) &$3.10 $& $4.02 $ & ($29$\%) & $3.94 $ & ($27$\%) & $3.29 $ & ($6.2$\%) & $3.30 $ & ($6.4$\%) & $2.88 $ & $1.3\times 10^{-3}$ \\
$(0.063$,$0.003$,$0.1$) &$11.0 $& $12.8 $ & ($16$\%) & $11.1 $ & ($0.83$\%) & $11.0 $ & ($0.38$\%) & $11.2 $ & ($1.7$\%) & $11.0 $ & $1.3\times 10^{-1}$
\end{tabular}
\caption{The cost of trading incurred for different values of $\varrho$, $\bar{\sigma}$, $\bar{\eta}$ and $\lambda$, for the one asset case ($n=1$) assuming coordinated variation.
The percentage values denote the increase compared to the continuous time optimal solution. 
Smaller percentages are better.
These results are based on $200$ realizations of $\xi^1$.  
	 }
\label{table:ex1ACV}
\end{table}
\end{landscape}


\begin{landscape}
\begin{table}[h]\footnotesize
\centering
\begin{tabular}{ l||c| c c|c c|c c|c c|c}
      &  Discrete  & \multicolumn{2}{c}{ } & \multicolumn{2}{c}{ } & \multicolumn{2}{c}{ } & \multicolumn{2}{c|}{ } & \\
    $(\bar{\sigma},\bar{\eta},\lambda)$ &  Optimal  & \multicolumn{2}{c}{CC} & \multicolumn{2}{c}{RHS} & \multicolumn{2}{c}{RHMC-I} & \multicolumn{2}{c|}{RHMC-II} & $\lambda\sigma^2/\eta$\\
\hline 
$\varrho=-20$\% & & \multicolumn{2}{c}{ } & \multicolumn{2}{c}{ } & \multicolumn{2}{c}{ } & \multicolumn{2}{c|}{ } \\
\hline 
($0.05$,$0.002$,$10^{-5}$) &$1.82 $ &$2.58 $ & ($42$\%) & $2.58 $ & ($42$\%) & $2.11 $ & ($18$\%) & $2.11 $ & ($18$\%) & $10^{-5}$ \\
($0.05$,$0.002$,$0.001$) &$1.86 $ &$2.62 $ & ($41$\%) & $2.61 $ & ($40$\%) & $2.15 $ & ($17$\%) & $2.15 $ & ($17$\%) & $10^{-3}$ \\
($0.05$,$0.002$,$0.1$) &$4.63 $ &$5.57 $ & ($20$\%) & $5.01 $ & ($8.4$\%) & $4.77 $ & ($3.2$\%) & $4.82 $ & ($4.7$\%) & $10^{-1}$ \\
($0.05$,$0.002$,$10^{-5}$) &$2.73 $ &$3.87 $ & ($42$\%) & $3.87 $ & ($42$\%) & $3.17 $ & ($18$\%) & $3.17 $ & ($18$\%) & $1.2\times 10^{-5}$ \\
($0.05$,$0.002$,$0.001$) &$2.77 $ &$3.91 $ & ($42$\%) & $3.90 $ & ($41$\%) & $3.20 $ & ($18$\%) & $3.20 $ & ($18$\%) & $1.2\times 10^{-3}$ \\
($0.05$,$0.002$,$0.1$) &$5.81 $ &$7.09 $ & ($22$\%) & $6.44 $ & ($10$\%) & $6.03 $ & ($3.9$\%) & $6.10 $ & ($5.3$\%) & $1.2\times 10^{-1}$ \\
($0.1$,$0.002$,$10^{-5}$) &$1.82 $ &$2.58 $ & ($42$\%) & $2.58 $ & ($42$\%) & $2.11 $ & ($18$\%) & $2.11 $ & ($18$\%) & $4\times 10^{-5}$ \\
($0.1$,$0.002$,$0.001$) &$1.98 $ &$2.73 $ & ($38$\%) & $2.68 $ & ($35$\%) & $2.25 $ & ($15$\%) & $2.25 $ & ($15$\%) & $4\times 10^{-3}$ \\
($0.1$,$0.002$,$0.1$) &$8.98 $ &$10.5 $ & ($17$\%) & $9.41 $ & ($4.7$\%) & $9.21 $ & ($2.7$\%) & $9.31 $ & ($4.2$\%) & $4\times 10^{-1}$ \\
($0.1$,$0.002$,$10^{-5}$) &$2.74 $ &$3.88 $ & ($42$\%) & $3.87 $ & ($42$\%) & $3.17 $ & ($18$\%) & $3.17 $ & ($18$\%) & $5\times 10^{-5}$ \\
($0.1$,$0.002$,$0.001$) &$2.89 $ &$4.03 $ & ($40$\%) & $3.98 $ & ($38$\%) & $3.30 $ & ($16$\%) & $3.31 $ & ($16$\%) & $5\times 10^{-3}$ \\
($0.1$,$0.002$,$0.1$) &$11.0 $ &$13.0 $ & ($18$\%) & $11.6 $ & ($5.4$\%) & $11.3 $ & ($2.7$\%) & $11.4 $ & ($4.2$\%) & $5\times 10^{-1}$ \\
\hline 
$\varrho=60$\% & & \multicolumn{2}{c}{ } & \multicolumn{2}{c}{ } & \multicolumn{2}{c}{ } & \multicolumn{2}{c|}{ } \\
\hline 
($0.05$,$0.002$,$10^{-5}$) &$1.82 $ &$2.58 $ & ($42$\%) & $2.58 $ & ($42$\%) & $2.11 $ & ($18$\%) & $2.11 $ & ($18$\%) & $10^{-5}$ \\
($0.05$,$0.002$,$0.001$) &$1.86 $ &$2.62 $ & ($41$\%) & $2.62 $ & ($41$\%) & $2.15 $ & ($18$\%) & $2.15 $ & ($18$\%) & $10^{-3}$ \\
($0.05$,$0.002$,$0.1$) &$4.77 $ &$5.64 $ & ($18$\%) & $5.45 $ & ($14$\%) & $5.02 $ & ($5.7$\%) & $5.04 $ & ($6.2$\%) & $10^{-1}$ \\
($0.05$,$0.002$,$10^{-5}$) &$2.73 $ &$3.87 $ & ($42$\%) & $3.87 $ & ($42$\%) & $3.17 $ & ($18$\%) & $3.17 $ & ($18$\%) & $1.2\times 10^{-5}$ \\
($0.05$,$0.002$,$0.001$) &$2.78 $ &$3.91 $ & ($42$\%) & $3.91 $ & ($42$\%) & $3.21 $ & ($18$\%) & $3.21 $ & ($18$\%) & $1.2\times 10^{-3}$ \\
($0.05$,$0.002$,$0.1$) &$5.98 $ &$7.18 $ & ($20$\%) & $6.95 $ & ($16$\%) & $6.31 $ & ($6.3$\%) & $6.34 $ & ($6.8$\%) & $1.2\times 10^{-1}$ \\
($0.1$,$0.002$,$10^{-5}$) &$1.82 $ &$2.58 $ & ($42$\%) & $2.58 $ & ($42$\%) & $2.11 $ & ($18$\%) & $2.11 $ & ($18$\%) & $4\times 10^{-5}$ \\
($0.1$,$0.002$,$0.001$) &$1.99 $ &$2.74 $ & ($38$\%) & $2.72 $ & ($37$\%) & $2.26 $ & ($16$\%) & $2.26 $ & ($16$\%) & $4\times 10^{-3}$ \\
($0.1$,$0.002$,$0.1$) &$9.31 $ &$10.6 $ & ($14$\%) & $10.2 $ & ($9.9$\%) & $9.79 $ & ($5.5$\%) & $9.82 $ & ($6.1$\%) & $4\times 10^{-1}$ \\
($0.1$,$0.002$,$10^{-5}$) &$2.74 $ &$3.88 $ & ($42$\%) & $3.88 $ & ($42$\%) & $3.17 $ & ($18$\%) & $3.17 $ & ($18$\%) & $5\times 10^{-5}$ \\
($0.1$,$0.002$,$0.001$) &$2.90 $ &$4.03 $ & ($40$\%) & $4.02 $ & ($39$\%) & $3.32 $ & ($16$\%) & $3.32 $ & ($17$\%) & $5\times 10^{-3}$ \\
($0.1$,$0.002$,$0.1$) &$11.4 $ &$13.1 $ & ($15$\%) & $12.6 $ & ($10$\%) & $12.0 $ & ($5.5$\%) & $12.0 $ & ($6.1$\%) & $5\times 10^{-1}$
\end{tabular}
\caption{The cost of trading incurred for different values of $\varrho$, $\bar{\sigma}$, $\bar{\eta}$ and $\lambda$, for the one asset case ($n=1$).
The percentage values denote the increase compared to the optimal solution. 
Smaller percentages are better.
These results are based on $200$ realizations of $\xi^1$ and $\xi^2$.  
	 }
\label{table:ex1A}
\end{table}
\end{landscape}


\begin{landscape}
\begin{table}[h]\footnotesize
\centering
\begin{tabular}{ l||c| c c|c c|c c|c c|c}
&  Discrete  & \multicolumn{2}{c}{ } & \multicolumn{2}{c}{ } & \multicolumn{2}{c}{ } & \multicolumn{2}{c}{ } &  \\
    $(\bar{\sigma},\bar{\eta},\lambda)$ &  Optimal  & \multicolumn{2}{c}{CC} & \multicolumn{2}{c}{RHS} & \multicolumn{2}{c}{RHMC-I} & \multicolumn{2}{c}{RHMC-II} & $\max\{\lambda\Xi^{-1}\Sigma\}$\\
\hline 
$\bsx_0=(100,100)^T$ & & \multicolumn{2}{c}{ } & \multicolumn{2}{c}{ } & \multicolumn{2}{c}{ } & \multicolumn{2}{c}{ } & \\
\hline 
($-0.8$,$0.002$,$10^{-5}$) &$4.23 $ &$5.46 $ & ($29$\%) & $5.46 $ & ($29$\%) & $4.63 $ & ($9.5$\%) & $4.63 $ & ($9.5$\%) & $4.6\times 10^{-8}$ \\
($-0.8$,$0.002$,$0.001$) &$4.28 $ &$5.51 $ & ($28$\%) & $5.47 $ & ($28$\%) & $4.69 $ & ($9.3$\%) & $4.68 $ & ($9.3$\%) & $4.6\times 10^{-6}$ \\
($-0.8$,$0.002$,$0.1$) &$8.15 $ &$12.6 $ & ($56$\%) & $10.0 $ & ($23$\%) & $8.74 $ & ($7.1$\%) & $8.64 $ & ($6.0$\%) & $4.6\times 10^{-4}$ \\
($-0.8$,$0.003$,$10^{-5}$) &$5.12 $ &$6.62 $ & ($29$\%) & $6.62 $ & ($29$\%) & $5.61 $ & ($9.6$\%) & $5.61 $ & ($9.6$\%) & $3\times 10^{-8}$ \\
($-0.8$,$0.003$,$0.001$) &$5.18 $ &$6.67 $ & ($29$\%) & $6.63 $ & ($28$\%) & $5.66 $ & ($9.4$\%) & $5.66 $ & ($9.4$\%) & $3\times 10^{-6}$ \\
($-0.8$,$0.003$,$0.1$) &$9.27 $ &$13.4 $ & ($45$\%) & $11.4 $ & ($22$\%) & $9.85 $ & ($6.2$\%) & $9.81 $ & ($5.8$\%) & $3\times 10^{-4}$ \\
($0.6$,$0.002$,$10^{-5}$) &$4.23 $ &$5.46 $ & ($29$\%) & $5.46 $ & ($29$\%) & $4.64 $ & ($9.5$\%) & $4.64 $ & ($9.5$\%) & $4.4\times 10^{-8}$ \\
($0.6$,$0.002$,$0.001$) &$4.53 $ &$5.76 $ & ($27$\%) & $5.71 $ & ($26$\%) & $4.91 $ & ($8.4$\%) & $4.91 $ & ($8.4$\%) & $4.4\times 10^{-6}$ \\
($0.6$,$0.002$,$0.1$) &$19.1 $ &$69.8 $ & ($268$\%) & $20.6 $ & ($7.9$\%) & $19.7 $ & ($2.7$\%) & $19.7 $ & ($2.7$\%) & $4.4\times 10^{-4}$ \\
($0.6$,$0.003$,$10^{-5}$) &$5.12 $ &$6.62 $ & ($29$\%) & $6.62 $ & ($29$\%) & $5.61 $ & ($9.6$\%) & $5.61 $ & ($9.6$\%) & $2.9\times 10^{-8}$ \\
($0.6$,$0.003$,$0.001$) &$5.42 $ &$6.92 $ & ($27$\%) & $6.88 $ & ($27$\%) & $5.89 $ & ($8.6$\%) & $5.89 $ & ($8.6$\%) & $2.9\times 10^{-6}$ \\
($0.6$,$0.003$,$0.1$) &$21.3 $ &$68.0 $ & ($222$\%) & $23.1 $ & ($8.4$\%) & $21.8 $ & ($2.5$\%) & $21.8 $ & ($2.5$\%) & $2.9\times 10^{-4}$ \\
\hline 
$\bsx_0=(100,-100)^T$ & & \multicolumn{2}{c}{ } & \multicolumn{2}{c}{ } & \multicolumn{2}{c}{ } & \multicolumn{2}{c}{ } & \\
\hline 
($-0.8$,$0.002$,$10^{-5}$) &$3.72 $ &$4.96 $ & ($33$\%) & $4.96 $ & ($33$\%) & $4.13 $ & ($11$\%) & $4.13 $ & ($11$\%) & $4.6\times 10^{-8}$ \\
($-0.8$,$0.002$,$0.001$) &$4.06 $ &$5.31 $ & ($31$\%) & $5.26 $ & ($29$\%) & $4.45 $ & ($9.6$\%) & $4.45 $ & ($9.6$\%) & $4.6\times 10^{-6}$ \\
($-0.8$,$0.002$,$0.1$) &$19.7 $ &$80.3 $ & ($312$\%) & $21.7 $ & ($10$\%) & $20.1 $ & ($2.1$\%) & $20.1 $ & ($2.0$\%) & $4.6\times 10^{-4}$ \\
($-0.8$,$0.003$,$10^{-5}$) &$4.63 $ &$6.13 $ & ($33$\%) & $6.13 $ & ($33$\%) & $5.12 $ & ($10$\%) & $5.12 $ & ($10$\%) & $3\times 10^{-8}$ \\
($-0.8$,$0.003$,$0.001$) &$4.96 $ &$6.48 $ & ($30$\%) & $6.42 $ & ($29$\%) & $5.43 $ & ($9.5$\%) & $5.43 $ & ($9.5$\%) & $3\times 10^{-6}$ \\
($-0.8$,$0.003$,$0.1$) &$22.1 $ &$78.1 $ & ($257$\%) & $24.2 $ & ($9.9$\%) & $22.5 $ & ($2.0$\%) & $22.5 $ & ($1.9$\%) & $3\times 10^{-4}$ \\
($0.6$,$0.002$,$10^{-5}$) &$3.72 $ &$4.96 $ & ($33$\%) & $4.96 $ & ($33$\%) & $4.13 $ & ($11$\%) & $4.13 $ & ($11$\%) & $4.4\times 10^{-8}$ \\
($0.6$,$0.002$,$0.001$) &$3.81 $ &$5.05 $ & ($32$\%) & $5.01 $ & ($31$\%) & $4.22 $ & ($10$\%) & $4.22 $ & ($10$\%) & $4.4\times 10^{-6}$ \\
($0.6$,$0.002$,$0.1$) &$9.93 $ &$19.4 $ & ($96$\%) & $13.2 $ & ($33$\%) & $10.4 $ & ($5.4$\%) & $10.5 $ & ($6.4$\%) & $4.4\times 10^{-4}$ \\
($0.6$,$0.003$,$10^{-5}$) &$4.62 $ &$6.13 $ & ($33$\%) & $6.13 $ & ($33$\%) & $5.12 $ & ($10$\%) & $5.12 $ & ($10$\%) & $2.9\times 10^{-8}$ \\
($0.6$,$0.003$,$0.001$) &$4.72 $ &$6.22 $ & ($32$\%) & $6.18 $ & ($31$\%) & $5.20 $ & ($10$\%) & $5.20 $ & ($10$\%) & $2.9\times 10^{-6}$ \\
($0.6$,$0.003$,$0.1$) &$11.2 $ &$19.6 $ & ($76$\%) & $14.7 $ & ($31$\%) & $11.7 $ & ($5.0$\%) & $11.9 $ & ($6.1$\%) & $2.9\times 10^{-4}$
\end{tabular}
\caption{The cost of trading incurred for different values of $\bsx_0$, $\rho$, $\bar{\eta}$ and $\lambda$, for the two asset case ($n=2$).
The percentage values denote the increase compared to the optimal solution. 
Smaller percentages are better.
These results are based on $200$ realizations of $\xi^1$ to $\xi^5$.  
	 }
\label{table:ex2A}
\end{table}
\end{landscape}

\clearpage

\bibliography{biblio}

\end{document}